\documentclass[journal]{IEEEtran}
\usepackage{xcolor}
\usepackage{pgf,tikz}
\usepackage{tabularx}
\usepackage{float}

\usepackage{amsthm}
\usepackage{graphicx}
\graphicspath{{../pdf.pdf/}{../jpeg/}{../eps.eps/}}
\DeclareGraphicsExtensions{.pdf,.jpeg,.png,.eps}

\usepackage{epstopdf}
\usepackage{multicol}
\usepackage{stfloats}
\usepackage{amsfonts}
\usepackage{mathrsfs}
\usepackage{color}
\usepackage[ruled]{algorithm2e}
\usepackage{fancybox}
\usepackage{framed}
\usepackage{psfrag, amssymb, amsmath, cite}
\usepackage{subfigure}
\usepackage{enumitem} 
\usepackage{capt-of}

\newtheorem{lem}{\textbf{Lemma}}

\newtheorem{rek}{\textbf{Remark}}

\newtheorem{definition}{\textbf{Definition}}
\newtheorem{theorem}{\textbf{Theorem}}
\newtheorem{corollary}{\textbf{Corollary}}
\newtheorem{assumption}{\textbf{Assumption}}
\newtheorem{simulation}{\textbf{Simulation}}

\newcommand{\vect}[1]{\boldsymbol{#1}}
\newcommand{\mat}[1]{\boldsymbol{#1}}

\newcommand{\wh}[1]{\widehat{#1}}
\newcommand{\RNum}[1]{\uppercase\expandafter{\romannumeral #1\relax}}

\def\diag{\mathop{\rm diag}}

\begin{document}

\title{Event-Triggered Algorithms for Leader-Follower Consensus of Networked Euler-Lagrange Agents}






\author{Qingchen~Liu,
		Mengbin Ye,
        Jiahu~Qin
        and~Changbin~Yu
\thanks{Q. Liu, M. Ye and C. Yu are with the Research School of Engineering, Australian National University, Canberra ACT 0200, Australia. C. Yu is also with the School of Automation, Hangzhou Dianzi University, Hangzhou 310018, China. {\tt\small \{qingchen.liu, mengbin.ye, brad.yu\}@anu.edu.au}.\newline
\indent J. Qin is with the Department of Automation, University of Science and Technology of China, Hefei 230027, China. {\tt\small jhqin@ustc.edu.cn}}}

\maketitle

\begin{abstract}
This paper proposes three different distributed event-triggered control algorithms to achieve leader-follower consensus for a network of Euler-Lagrange agents. We firstly propose two model-independent algorithms for a subclass of Euler-Lagrange agents without the vector of gravitational potential forces. By model-independent, we mean that each agent can execute its algorithm with no knowledge of the agent self-dynamics. A variable-gain algorithm is employed when the sensing graph is undirected; algorithm parameters are selected in a fully distributed manner with much greater flexibility compared to all previous work concerning event-triggered consensus problems. When the sensing graph is directed, a constant-gain algorithm is employed. The control gains must be centrally designed to exceed several lower bounding inequalities which require limited knowledge of bounds on the matrices describing the agent dynamics, bounds on network topology information and bounds on the initial conditions. When the Euler-Lagrange agents have dynamics which include the vector of gravitational potential forces, an adaptive algorithm is proposed which requires more information about the agent dynamics but can estimate uncertain agent parameters.

For each algorithm, a trigger function is proposed to govern the event update times. At each event, the controller is updated, which ensures that the control input is piecewise constant and saves energy resources. We analyse each controllers and trigger function and exclude Zeno behaviour. Extensive simulations show $1)$ the advantages of our proposed trigger function as compared to those in existing literature, and $2)$ the effectiveness of our proposed controllers.
\end{abstract}

\section{Introduction}\label{sec1}
The field of multi-agent systems has received extensive attention from the control community in the past two decades. In particular, coordination of a network of interacting agents to achieve a global objective has been seen as a key sub-area within the field. See \cite{cao2013} for a recent survey. Leader-follower consensus is a variation of the commonly studied consensus problem where, with all agents having a commonly defined state variable(s), the network of follower agents converge to the state values of the stationary leader. This is achieved by interaction between neighbouring agents. By the fact that interaction is only between neighbouring agents then each follower agent must use a distributed controller, i.e. agents cannot use global information about the whole  network \cite{ren2008distributed}. 

The Euler-Lagrange equations describe the dynamics of a large class of nonlinear systems (including many mechanical systems such as robotic manipulators, spacecraft and marine vessels) \cite{kelly2006control, ortega2013passivity_book}. As a result, there is motivation to study multi-agent coordination problems where each agent has Euler-Lagrange dynamics \cite{ren2011distributed_book}. Leader-follower consensus for directed networks of Euler-Lagrange agents has been studied in \cite{ye2015ELrendezvous} using a model-independent controller, and in \cite{mei2013distributed} using an adaptive controller. In both, the topology requires a directed spanning tree.

Recently, use of event-triggered controllers has been popularised in multi-agent coordination problems \cite{qin2016recent}. While each agent has continuous time dynamics, the controller is updated at discrete time instants based on event-scheduling. Because the controller updates occur at specific events, this has the benefit of reducing  actuator updates. However, it is important to properly design and analyse the event-scheduling trigger function to exclude Zeno behaviour \cite{zhang2001zeno, ames2006stability}, which can cause the controller to collapse. Numerous results have been published studying consensus based problems using distributed event-triggered control laws. However, the majority study agents with single and double-integrator dynamics \cite{dimarogonas2012distributed, seyboth2013event, fan2013distributed, Fan015ZenoFree, liu2016QuantizedEvent, Cameron2016EventTrigger}.

There have been relatively few results published studying event-triggered control for networks of Euler-Lagrange agents. Pioneering contributions studied leaderless consensus (but not leader-follower consensus) on an undirected network \cite{mu2014consensus, huang2016distributed}. The dynamics studied in \cite{mu2014consensus} and \cite{huang2016distributed} are a subclass of Euler-Lagrange dynamics as they do not consider the presence of gravitational forces for each agent. While continuous model-independent algorithms e.g. \cite{ye2015ELrendezvous} are easily adapted to be event-triggered, as shown in \cite{mu2014consensus, huang2016distributed}, they cannot guarantee the coordination objective in the presence of gravitational forces (which has an effect which is similar to a bounded disturbance). Typical control techniques required to deal with this term include feedback linearisation \cite{meng2012FT_EL_track}, adaptive control \cite{mei2013distributed} and sliding mode control \cite{mei2011distributed_EL}. We note that these techniques have not been well studied in an event-triggered framework. In fact, first-order sliding mode controllers exhibit Zeno-like behaviour and thus are unsuitable for implementation in event-triggered control. In \cite{liu2016decentralised}, an adaptive, event-triggered controller is proposed to achieve flocking behaviour for undirected networks of Euler-Lagrange agents. This allows for gravitational forces omitted in \cite{mu2014consensus, huang2016distributed}. However, the proposed controller is piecewise continuous, which restricts its implementation in digital platforms. Moreover, it is worth noting that the trigger function used in \cite{liu2016decentralised} cannot eliminate Zeno behaviour for each agent.

\subsection{Contributions of This Paper}

In this paper, we propose three different distributed event-triggered control algorithms to achieve leader-follower consensus for networked Euler-Lagrange agents; each algorithm has different strengths and their appropriateness of use may depend on the application scenario. 

We propose two model-independent controllers for Euler-Lagrange agents without the gravitational term (unlike the model-dependent algorithm in \cite{mu2014consensus}). Firstly, a globally asymptotically stable variable-gain model-independent algorithm is proposed for agents on undirected graphs. The variable-gain controller allows for fully distributed \emph{and arbitrary design of parameters in \emph{both the control algorithm and trigger function}. For agents with complex dynamics, almost all existing results require centralised design of key parameters in the trigger function using limited global knowledge of the network \cite{mu2014consensus, liu2016decentralised, huang2016EL_eventtrigger}. The design of these key parameters is to ensure either Zeno-free behaviour, or to guarantee convergence of the controller. In the case of simple agent dynamics, the parameters are distributed in design but must either obey upper or lower bounds {\cite{dimarogonas2012distributed, fan2013distributed,garcia2013decentralised,seyboth2013event}}.} As such, the fully distributed variable-gain controller represents a significant advance on existing event-triggered control algorithms, because stability and convergence are always guaranteed, even if the algorithm and trigger function parameters are arbitrarily selected.

Even when implemented continuously, and with simple agent dynamics, variable-gain algorithms on directed graphs are difficult to analyse \cite{mei2014variable_linear,mei2016adaptive,li2014cooperative_book}. For the second model-independent controller, which is applicable for directed graphs, we are therefore motivated to use constant control gains. It will become apparent in the sequel that, even when the controller has constant gains, the combination of Euler-Lagrange dynamics, directed topology, and event-based control requires nontrival stability analysis. The algorithm achieves leader-follower consensus semi-globally, exponentially fast (neither directed graphs nor exponential stability has not been established in any existing results on event-based networks of Euler-Lagrange agents). Some limited knowledge of the bounds on the agent dynamic parameters, the network topology and a set of all possible initial conditions is required to centrally design the control gains. This is a trade-off for allowing agents  to interact on a directed graph. 

Lastly, we propose a globally asymptotically stable adaptive algorithm for use when the gravitational term is present in the agent self-dynamics; this algorithm appeared in our preliminary work \cite{liu2016ELevent_trigger}. The adaptive algorithm is able to estimate uncertain dynamical parameters, but requires increased knowledge about the agent self-dynamics.

All three proposed controllers are piecewise constant (unlike the piecewise continuous algorithm in \cite{liu2016decentralised}), which has the benefit of reducing actuator updates and thus conserving energy resources. Furthermore, each agent only requires state, and relative state measurements, and does not require knowledge of the trigger times of neighbouring agents (unlike \cite{mu2014consensus, liu2016decentralised}). For each algorithm, a trigger function is proposed and we show that Zeno behaviour can be excluded for every agent by proving that for any finite interval of time, a strictly positive lower bound exists on the time between each event. All three trigger functions are of the same form with only minor modifications. Each term of the trigger function is carefully selected to ensure that the trigger function is more effective, when compared with existing trigger functions which do one or the other, at $1)$ reducing the total number of events, and $2)$ eliminating Zeno behaviour for every agent. We show this by detailed comparison and analysis based on simulations. As a result of having multiple terms in the trigger function to achieve the aforementioned improvements, the stability analysis is made significantly more complicated. Each algorithm requires a different approach to proving stability, and the proposed methods may be useful for other problems in event-based control of multi-agent systems.

\subsection{Structure of the Rest of the Paper}
Section~\ref{sec:background} provides mathematical notations and background on graph theory and Euler-Lagrange systems. A formal problem definition is also provided. The three different distributed event-triggered control algorithms are then proposed and analysed in Section~\ref{sec:undir_MI}, \ref{sec:dir_IC} and \ref{sec:adaptive}, separately. Simulations in Section~\ref{sec:simulation} show the effectiveness of the proposed controllers. Concluding remarks are given in Section~\ref{sec:conclusion}.

\section{Background and problem statement}\label{sec:background}
\subsection{Notations and Mathematical Preliminaries}
In this paper, $\mathbb{R}^{n}$ denotes the $n$-dimensional Euclidean space and $\mathbb{R}^{m\times n}$ denotes the set of $m\times n$ real matrices.  The transpose of a vector or matrix $\mat{A}$ is given by $\mat{A}^\top$. The $i^{th}$ smallest eigenvalue of a symmetric matrix $\mat{A}$ is denoted by $\lambda_{i}(\mat{A})$. Let $\mat{x} = [\mat{x}_1,\ldots,\mat{x}_n]^\top$ where $\mat{x}_i\in\mathbb{R}^{n\times n}$ and $n \geq 1$. Then $diag\{\mat{x}\}$ denotes a (block) diagonal matrix with the (block) elements of $\mat{x}$ on its diagonal, i.e. $diag\{\mat{x}_1, ..., \mat{x}_n\}$. A symmetric matrix $\mat{A}\in\mathbb{R}^{n\times n}$ which is positive definite (respectively nonnegative definite) is denoted by $\mat{A} > 0$ (respectively $\mat A \geq 0$).  For two symmetric matrices $\mat{A}, \mat{B}$, the expression $\mat{A} > \mat{B}$ is equivalent to \mbox{$\mat{A} - \mat{B} > 0$}. The $n\times n$ identity matrix is $\emph{I}_{n}$ and $\mathbf{1}_n$ denotes an $n$-tuple column vector of all ones. The $n\times 1$ column vector of all zeros is denoted by $\mathbf{0}_n$. The symbol $\otimes$ denotes the Kronecker product. The Euclidean norm of a vector, and the matrix norm induced by the Euclidean norm, are denoted by $\|\cdot\|$. The absolute value of a real number is $|\cdot |$. For the space of piecewise continuous, bounded vector functions, the norm is defined as $\|f\|_{\mathcal{L}_\infty}=\sup\|f(t)\|<\infty$ and the space is denoted by $\mathcal{L}_\infty$. The space $\mathcal{L}_p$ for $1\leq p<\infty$ is defined as the set of all piecewise continuous vector functions such that $\|f\|_{\mathcal{L}_p}=\left(\int_0^\infty\|f(t)\|^pdt\right)^{1/p}<\infty$ where $p$ refers to the type of $p$-norm.

Several theorems, lemmas and corollaries are now introduced, which will be used in this paper.

\begin{theorem}[Mean Value Theorem for Vector-Valued Functions  \cite{rudin1964principles}]\label{theorem:mean_value_vector}
For a continuous vector-valued function $\vect{f}(s):\mathbb{R}\rightarrow\mathbb{R}^n$ differentiable on $s\in [a,b]$, there exists $t\in(a,b)$ such that 
\begin{align*}
\left\|\frac{d\vect f}{ds}(t)\right\|\geq\frac{1}{b-a}\|\vect{f}(b)-\vect{f}(a)\|
\end{align*}
\end{theorem}

\begin{theorem}[The Schur Complement \cite{horn2012matrixbook}]\label{theorem:schur_complement} 
Consider a symmetric block matrix, partitioned as 
\begin{equation}\label{eq:schur_def_A}
\mat{A} = \begin{bmatrix}\mat{B} & \mat{C} \\ \mat{C}^\top & \mat{D}\end{bmatrix}
\end{equation}
Then $\mat{A} > 0 $  if and only if $\mat{B} > 0$ and $\mat{D} - \mat{C}^\top\mat{B}^{-1}\mat{C} > 0$. Equivalently, $\mat A > 0$ if and only if $\mat D > 0$ and $\mat B - \mat C \mat D^{-1} \mat C^\top > 0$.
\end{theorem}

\begin{lem}(From \cite{ioannou2006adaptive})\label{lem:l1}
If a function $f(t)$ satisfies $f(t), \dot{f}(t)\in L_\infty$, and $f(t)\in L_p$ for some value of $p\in[1,\infty)$, then $f(t)\rightarrow 0$ as $t\rightarrow\infty$.
\end{lem}

\begin{lem}\label{lem:W_bound}
Suppose $\mat A > 0$ is defined as in \eqref{eq:schur_def_A}. Let a quadratic function with arguments $\vect x, \vect y$ be expressed as $W = [\vect x^\top, \vect y^\top] \mat A [\vect x^\top, \vect y^\top]^\top$. Define $\mat F: = \mat B - \mat C \mat D^{-1} \mat C^\top $ and $\mat G := \mat{D}-\mat{C}^\top\mat{B}^{-1}\mat{C}$. Then there holds
\begin{subequations}\label{eq:W_bound_xy}
\begin{align}
\lambda_{\min}(\mat F) \vect x^\top \vect x \leq  \vect{x}^\top\mat F\vect{x}  & \leq W \label{eq:W_bound_x} \\
\lambda_{\min}(\mat G) \vect y^\top \vect y \leq  \vect{y}^\top\mat G\vect{y}  & \leq W \label{eq:W_bound_y} 
\end{align}
\end{subequations}
\begin{proof} The proof for \eqref{eq:W_bound_y} is immediately obtained by recalling \textbf{Theorem}~\ref{theorem:schur_complement} and observing that 
\begin{align*}
W &= \vect{y}^\top\mat G\vect{y}  + [\vect{y}^\top\mat{C}^\top\mat{B}^{-1}+\vect{x}^\top]\mat{B}[\mat{B}^{-1}\mat{C}\vect{y}+\vect{x}]
\end{align*}
An equally straightforward proof yields \eqref{eq:W_bound_x}.
\end{proof}
\end{lem}

\begin{lem}\label{lem:bounded_pd_function}
Let $g(x,y)$ be a function given as
\begin{equation}
g(x,y) = a x^2 + b y^2 - c x y^2 - d x y
\end{equation} 
for real positive scalars $a,c,d > 0$. Then for a given $\mathcal{X} > 0$, there exist $b > 0$ such that $g(x,y) > 0$ for all $y \in [0, \infty)$ and $x \in [0,\mathcal{X}]$.
\begin{proof}
See Appendix \ref{app:sec_background}.
\end{proof}
\end{lem}

\begin{corollary}\label{cor:bounded_pd_function_2}
Let $h(x,y)$ be a function given as
\begin{equation}
h(x,y) = a x^2 + b y^2 - c x y^2 - d x y - e x - f y 
\end{equation} 
where the real, strictly positive scalars $c,d,e,f$ and two further positive scalars $\varepsilon, \vartheta$ are fixed. Suppose that for given $\mathcal{Y}, \varepsilon$ there holds $\mathcal{Y}- \varepsilon > 0$, and for a given $\mathcal{X} > 0$ there holds $\mathcal{X}-\vartheta > 0$. Define the sets $\mathcal{U} = \{x,y : x\in [\mathcal{X}-\vartheta, \mathcal{X}], y > 0\}$ and $\mathcal{V} = \{x, y : x > 0, y\in [\mathcal{Y}-\varepsilon,\mathcal{Y}]\}$. Define the region $\mathcal{R} = \mathcal{U} \cup \mathcal{V}$. Then there exist $a, b > 0$ such that $h(x,y)$ is positive definite in $\mathcal{R}$.
\begin{proof}
See Appendix \ref{app:sec_background}.
\end{proof}
\end{corollary}

\subsection{Graph Theory}
We model the interactions among the leader and $n$ followers by a weighted directed graph (digraph) $\mathcal{G}=(\mathcal{V},\mathcal{E},\mathcal{A})$ with vertex set $\mathcal{V}=\{v_0,v_1,\cdots,v_n\}$ and edge set $\mathcal{E}\subseteq \mathcal{V}\times \mathcal{V}$. Without loss of generality, the leader agent is numbered by $v_0$. We use $\mathcal{G}_F$ to describe the interactions among the $n$ follower agents with vertex set $\mathcal{V}_F=\{v_1,\cdots,v_n\}$ and edge set $\mathcal{E}_F\subseteq \mathcal{V}_F\times \mathcal{V}_F$. An ordered edge set of $\mathcal{G}$ is $e_{ij}=(v_i,v_j)$. The weighted adjacency matrix $\mathcal{A}=\mathcal{A}(\mathcal{G})=\{a_{ij}\}$ is the $(n+1)\times (n+1)$ matrix given by $a_{ij}>0$, if $e_{ji}\in \mathcal{E}$ and $a_{ij}=0$, otherwise. In this paper, it is assumed that $a_{ii} = 0$, i.e. there are no self-loops. The edge $e_{ij}$ is incoming with respect to $v_j$ and outgoing with respect to $v_i$. A graph is undirected if $e_{ij}\in\mathcal{E} \Leftrightarrow e_{ji}\in\mathcal{E}$ and $a_{ij} = a_{ji}$. The neighbour set of $v_i$ is denoted by $\mathcal{N}_i=\{v_j\in\mathcal{V}:(v_i,v_j)\in\mathcal{E}\}$. The $(n+1)\times (n+1)$ Laplacian matrix, $\mathcal{L}=\{l_{ij}\}$, of the associated directed graph $\mathcal{G}$ is defined as
\begin{equation*}
l_{ij}=
\begin{cases}
\sum_{k=1,k\neq i}^{n}a_{ik} & \text{for}~~j=i\\
-a_{ij} & \text{for}~~j\neq i
\end{cases}
\end{equation*}
A digraph with $n+1$ vertices is called a directed spanning tree if it has $n$ edges and there exists a root vertex with directed paths to every other vertex \cite{ren2008distributed}. The following result holds for the Laplacian matrix associated with a directed graph.

\begin{lem}(From \cite{ren2008distributed})
Let $\mathcal{L}$ be the Laplacian matrix associated with a directed graph. Then $\mathcal{L}$ has a simple zero eigenvalue and all other eigenvalues have positive real parts if and only if $\mathcal{G}$ has a directed spanning tree.
\end{lem}

\begin{lem}(From \cite{song2012pinning})\label{lem:lem5}
Suppose a graph $\mathcal{G}$ contains a directed spanning tree, and there are no edges of $\mathcal{G}$ which are incoming to the root vertex $v_0$ of the tree. Then the Laplacian matrix associated with $\mathcal{G}$ has the following form:
\begin{equation*}
\mathcal{L}=
\begin{bmatrix}
0 & \mathbf{0}_{n-1}^T\\
\mathcal{L}_{21} & \mathcal{L}_{22}\\
\end{bmatrix}
\end{equation*}
and all eigenvalues of $\mathcal{L}_{22}$ have positive real parts. Moreover, there exists a diagonal positive definite matrix $\mat \Gamma$ such that $\mat Q:= \mat \Gamma \mathcal{L}_{22} + \mathcal{L}_{22}^\top \mat \Gamma > 0$. In addition, if $\mathcal{G}_F$ is undirected, then $\mathcal{L}_{22}$ is symmetric positive definite. 
\end{lem}

\subsection{Euler-Lagrange Systems}
A class of dynamical systems can be described using the Euler-Lagrange equations \cite{kelly2006control}. The general form for the $i^{th}$ agent equation of motion is:
\begin{align}\label{eq:el_dynamics_1}
\mat{M}_{i}(\vect{q}_i)\ddot{\vect{q}}_i+\mat{C}_i(\vect{q}_i,\dot{\vect{q}}_{i})\dot{\vect{q}}_{i} + \vect g_i(\vect q_i) =\vect{\tau}_i   
\end{align}
where $\vect{q}_i\in\mathbb{R}^p$ is a vector of the generalized coordinates, $\mat{M}_i(\vect{q}_i)\in\mathbb{R}^{p\times p}$ is the inertial matrix, $\mat{C}_i(\vect{q}_i,\dot{\vect{q}}_{i})\in\mathbb{R}^{p\times p}$ is the Coriolis and centrifugal torque matrix, $\vect g_i(\vect q_i) \in \mathbb{R}^p$ is the vector of gravitational forces and $\vect{\tau}_{i}\in\mathbb{R}^{p}$ is the control input vector. For agent $i$, we have $\vect{q}_i=[\vect{q}_i^{(1)},\ldots,\vect{q}_i^{(p)}]^\top$. We assume each agent is fully actuated. The dynamics in \eqref{eq:el_dynamics_1} are assumed to satisfy the following properties, details of which are provided in \cite{kelly2006control}.
\begin{enumerate}[label=P\arabic*]
\item \label{prty:EL_M_symPD} The matrix $\mat{M}_i(\vect{q}_i)$ is symmetric positive definite. 
\item \label{prty:EL_M_bound} There exist constants $k_{\underline{m}}, k_{\overline{M}} > 0$ such that $k_{\underline{m}}\mat{I}_p\leq \mat{M}_i(\vect{q}_i)\leq k_{\overline{M}}\mat{I}_p, \forall\,i,\vect q_i$. It follows that $\sup_{\vect{q}_i}\Vert\mat{M}_i\Vert_2 \leq k_{\overline{M}}$ and $k_{\underline{m}}\leq \inf_{\vect{q}_i}{ \Vert\mat{M}_i^{-1}\Vert_2}^{-1}$ $\forall\,i$.
\item \label{prty:EL_C_bound} There exists a constant $k_{C} > 0$ such that \mbox{$\Vert\mat{C}_i\Vert_2 \leq k_{C} \Vert\dot{\vect{q}}_i\Vert_2,\forall\,i,\dot{\vect q}_i$}.
\item \label{prty:EL_CM_skew} The matrix $\mat{C}_{i}(\vect{q}_i,\dot{\vect{q}}_i)$ is related to the inertial matrix $\mat{M}_{i}(\vect{q}_i,\dot{\vect{q}}_i)$ by the expression $\vect{x}^T(\frac{1}{2}\dot{\mat{M}}_{i}(\vect{q}_i)-\mat{C}_{i}(\vect{q}_i,\dot{\vect{q}}_i))\vect{x}=0$ for any $\vect{q},\dot{\vect{q}},\vect{x}\in\mathbb{R}^p$. This implies that $\dot{\mat{M}}_{i}(\vect{q}_i) = \mat{C}_{i}(\vect{q}_i,\dot{\vect{q}}_i) + \mat{C}_{i}(\vect{q}_i,\dot{\vect{q}}_i)^\top$.
\item \label{prty:EL_g_bound} There exists a constant $k_g>0$ such that $\Vert \vect{g}_{i}(\vect{q}_i)\Vert < k_g$.
\item \label{prty:EL_Ytheta} Linearity in the parameters: $\mat{M}_{i}(\vect{q}_{i})\vect{x}+\mat{C}_{i}(\vect{q}_{i},\dot{\vect{q}}_{i})\vect{y}+\vect{g}_{i}(\vect{q}_{i})=\mat{Y}_{i}(\vect{q}_{i},\dot{\vect{q}}_{i},\vect{x},\vect{y})\vect{\Theta}_{i}$ for all vectors $x,y\in\mathbb{R}^{p}$, where $\mat{Y}_{i}(\vect{q}_{i},\dot{\vect{q}}_{i},\vect{x},\vect{y})$ is the known regressor matrix and $\mat{\Theta}_{i}$ is a vector of unknown but constant parameters associated with the $i^{th}$ agent.
\end{enumerate}

\begin{rek}
It is assumed throughout this paper that the properties \ref{prty:EL_M_symPD} through to \ref{prty:EL_Ytheta} always hold.
\end{rek}

\begin{assumption}[Sub-class of dynamics]\label{assm:no_gravity}
In Sections~\ref{sec:undir_MI} and \ref{sec:dir_IC}, we assume that $\vect g_i(\vect q_i) = \vect 0, \forall\, i$. In other words, the dynamics of the agents belong to a subclass of Euler-Lagrange equations which do not have a gravity term. That is,
\begin{align}\label{eq:el_dynamics_2}
\mat{M}_{i}(\vect{q}_i)\ddot{\vect{q}}_i+\mat{C}_i(\vect{q}_i,\dot{\vect{q}}_{i})\dot{\vect{q}}_{i} =\vect{\tau}_i   
\end{align}
\end{assumption}

If the gravity term $\vect g_i(\vect q_i)$ is present, the adaptive controller proposed in Section~\ref{sec:adaptive} may be used.

\subsection{Problem Statement}
Denote the leader as agent 0 with $\vect{q}_0$ and $\dot{\vect{q}}_0$ being the generalised coordinates and generalised velocity of the leader, respectively. The aim is to develop event-based, distributed algorithms for each Euler-Lagrange follower agent, where the updates are such that $\tau_i$ is piecewise-constant. The distributed algorithms are designed to achieve leader-follower consensus to a stationary leader, i.e. $\dot{\vect{q}}_0(t) = 0, \forall t \geq 0$. Leader-follower consensus is said to be achieved if $\lim_{t\rightarrow\infty}\|\vect{q}_i(t)-\vect{q}_0(t)\|=0, \forall i=1,\ldots,n$ and $\lim_{t\rightarrow\infty}\|\dot{\vect{q}}_i(t)\|=0, \forall i=1,\ldots,n$ are satisfied.

Another aim of this paper is to exclude the possibility of Zeno behaviour. We provide a formal definition of Zeno behaviour in the sequel. Zeno behaviour of an event-based controller means an infinite number of controller updates occur in a finite time period, which is undesirable since no practical controller can do this.

In this paper, we assume that agent $i\in{1,\ldots,n}$ is equipped with sensors which continuously measures the relative generalised coordinates to agent $i$'s neighbours. In other words, $\vect q_i(t)-\vect q_j(t), \forall j\in\mathcal{N}_i$ is available to agent $i$. In Section~\ref{sec:dir_IC} we also assume that the relative generalised velocities are available, i.e. $\dot{\vect q}_i(t)-\dot{\vect q}_j(t), \forall j\in\mathcal{N}_i$. The scenario where agents collect relative information to execute algorithms can be found in many experimental testbeds, such as ground robots or UAVs equipped with high-speed cameras. It is also assumed that each agent $i$ can measure its own generalised velocity continuously, $\dot{\vect q}_i(t)$.

\begin{definition}\label{def:1}
Let a finite time interval be $t_\mathcal{Z} = [a,b]$ where $0 \leq a < b < \infty$. If, for some finite $k\geq 0$, the sequence of event triggers $\{t_k^i, ..., t_\infty^i\} \in [a,b]$ then the system exhibits Zeno behaviour. 
\end{definition}

\section{Main result: a Variable-gain, model-independent controller on undirected networks}\label{sec:undir_MI}
In this section, we introduce a variable gain, event-triggered control algorithm to achieve leader-follower consensus for Euler-lagrange agents where the network model of the follower agents is described by an undirected network. We show that the proposed algorithm do not require any knowledge of the multi-agent system (i.e totally distributed design) and is globally stable. Zeno behaviour is also excluded for each agent in the system.
\subsection{Main Result}
Define a new state variable for agent $i$ as
\begin{align*}
\vect{z}_i(t)=\sum_{j=0}^{n}a_{ij}(\vect{q}_i(t)-\vect{q}_j(t))+\mu_i(t)\dot{\vect{q}}_i(t),~~~~~i=1,\ldots,n 
\end{align*}
where $a_{ij}$ is an element of the adjacency matrix $\mathcal{A}$ associated with the digraph $\mathcal{G}$. Note that the follower graph $\mathcal{G}_F$ is undirected. We let $\mu_i(t)$ be subject to the following updating law:
\begin{align}
\dot{\mu}_i(t)& = \alpha \dot{\vect{q}}_i^\top(t) \dot{\vect{q}}_i(t) \label{eq:control_gain}
\end{align}
The scalar $\alpha$ is strictly positive and is universal for all agents. It is obvious that $\mu_i(t)$ is a monotonically increasing function. The variable-gain scalar function $\mu_i(t)$ is initialised at $t=0$ with an arbitrary $\mu_i(0) \geq 0$, which implies that $\mu_i(t) \geq 0,\forall\,t > 0$. 

The control algorithm is now proposed. Let the trigger time sequence of agent $i$ be $t_{0}^{i}, t_{1}^{i}, \ldots, t_{k}^{i},\ldots$ with \mbox{$t_0^i := 0$} and we detail below how each trigger time is determined. The event-triggered controller for follower agent $i$ is designed as:
\begin{align}
\vect{\tau}_i(t)&=-\vect{z}_i(t_k^i)   \label{eq:control_input_1}
\end{align}
for $t\in[t_k^i,t_{k+1}^i)$. The control input for each agent is held constant and equal to the last control update $\vect{\tau}_i(t_k^i)$ in the time interval $[t_k^i, t_{k+1}^i)$.

We define a state mismatch for agent $i$ between consecutive event times $t_k^i$ and $t_{k+1}^i$ as follows:
\begin{align}
\vect{e}_i(t)&=\vect{z}_i(t_k^i)-\vect{z}_i(t)  \label{eq:measurement_error}
\end{align}
for $t\in[t_k^i,t_{k+1}^i)$. Then we design the trigger function as follows:
\begin{align}\label{eq:trigger_function_1}
f_i(\vect{e}_i,\dot{\vect{q}}_i,\omega_i)=\|\vect{e}_i(t)\|^2-\beta_i\|\dot{\vect{q}}_i(t)\|^2-\omega_i(t) 
\end{align}
where $\beta_i$ is \emph{an arbitrarily chosen positive constant} (see the Proof of \textbf{Theorem} \ref{theorem_2} for the explanations), $\omega_i(t)$ is an offset function defined as $\omega_i(t)=\kappa_i\exp(-\varepsilon_it)$ with arbitrarily chosen $\kappa_i, \varepsilon_i>0$. The $k^{\text{th}}$ event for agent $i$ is triggered as soon as the trigger condition $f_i(\vect{e}_i,\dot{\vect{q}}_i,\omega_i)=0$ is satisfied. The control input $\vect{\tau}_i(t)$ is updated only when an event of agent $i$ is triggered. Furthermore, every time an event is triggered, and in accordance with their definitions, the measurement error $\vect e_i(t)$ is reset to be equal to zero and thus the trigger function assumes a non-positive value, that is, $f_i(\vect{e}_i,\dot{\vect{q}}_i,\omega_i)\leq 0$.

\begin{rek}
In existing event-based multi-agent control literature, the parameters of the state-dependent term are typically restricted. For example, the authors of \cite{huang2016distributed} studied an event-triggered controller which achieved leaderless consensus for networked Euler-Lagrange agents under an undirected graph. Different from our proposed variable-gain controller, their controller adopts fixed gains. As a result, the parameter $\varrho_i$ (see the trigger function in \cite{huang2016distributed}) of the state-dependent term has to be less than a computable upper bound. This bound requires knowledge of the control gains and graph topology (e.g. number of neighbours and degree of the agent). In comparison, our equivalent parameter $\beta_i$ in our proposed trigger function \eqref{eq:trigger_function_1} can be chosen as an arbitrary positive constant. This provides a much greater flexibility in the implementation of the algorithm.

We note that even in papers considering simple single integrator dynamics with a parameter for the state-dependent term, equivalent to our $\beta_i$, require an upper bound as well (see the seminal works of \cite{dimarogonas2012distributed,fan2013distributed}). To the best of the authors' knowledge, the event-based controller proposed in this paper is the first one to allow an arbitrarily chosen positive parameter for the state-dependent term in the trigger function.
\end{rek}

By substituting the control input \eqref{eq:control_input_1} into the system dynamics \eqref{eq:el_dynamics_2}, the closed-loop system can be written as
\begin{align}\label{eq:system_dynamic_1}
\mat{M}_{i}(\vect{q}_i)\ddot{\vect{q}}_i(t)+\mat{C}_i(\vect{q}_i,\dot{\vect{q}}_{i})\dot{\vect{q}}_{i}(t)=-\vect{z}_i(t_k^i)
\end{align}
Then by applying $\eqref{eq:measurement_error}$, we obtain
\begin{align}\label{eq:system_dynamic_2}
\mat{M}_{i}(\vect{q}_i)\ddot{\vect{q}}_i(t)+\mat{C}_i(\vect{q}_i,\dot{\vect{q}}_{i})\dot{\vect{q}}_{i}(t)=-(\vect{z}_i(t)+\vect{e}_i(t))
\end{align}
Define new state variables $\vect{u}_i = \vect{q}_{i} - \vect{q}_0$ and $\vect{v}_i = \dot{\vect{q}}_{i}$ and we drop the argument $t$ for brevity, and where there is no confusion. Define the stacked column vectors of all $\vect u_i, \vect v_i, \vect q_i, \vect e_i$ as $\vect{u} = [\vect{u}_1^\top, ..., \vect{u}_{n}^\top]^\top$, $\vect{v} = [\vect{v}_1^\top, ..., \vect{v}_{n}^\top]^\top$, $\vect{q} = [\vect{q}_{1}^\top, ..., \vect{q}_{n}^\top]^\top$, $\vect{z}= [\vect{z}_1^\top, ..., \vect{z}_{n}^\top]^\top$ and $\vect{e} = [\vect{e}_1^\top, ..., \vect{e}_{n}^\top]^\top$ respectively. It is easy to obtain that
\begin{align*}
\vect{z}&=(\mathcal{\mat L}_{22}\otimes\vect{I}_p)(\vect{q}-\vect{1}_n\otimes\vect{q}_0)+\mat K\dot{\vect q}\\
		&=(\mathcal{\mat L}_{22}\otimes\vect{I}_p)\vect{u}+\mat K\vect v
\end{align*}
where $\mat K = \diag[\mu_1 \mat I_p, ..., \mu_n \mat I_p]$. Define the following block diagonal matrices $\mat{M}(\vect{q}) = \diag[\mat{M}_1(\vect{q}_1), ..., \mat{M}_n(\vect{q}_n)]$, $\mat{C}(\vect{q},\dot{\vect{q}}) = \diag[\mat{C}_1(\vect{q}_1, \dot{\vect{q}}_1), ..., \mat{C}_n(\vect{q}_n, \dot{\vect{q}}_n)]$. It is obvious that $\mat{M}$ is symmetric positive definite since $\mat{M}_i > 0, \forall\,i $. With these notations, the compact form of system \eqref{eq:system_dynamic_2} can be expressed as
\begin{align}\label{eq:compact_system_dynamic_1}
&\dot{\vect{u}}  = \vect{v}\notag\\
&\dot{\vect{v}}  = -\mat{M}(\vect{q})^{-1}\left[\mat{C}(\vect{q},\vect{v})\vect{v}+(\mathcal{\mat L}_{22} \otimes \mat{I}_p)\vect{u}+\mat{K}\vect{v}+\vect{e}\right]\notag\\
&\dot{\mat K}  = \alpha (\mat \Xi \otimes \mat I_p) 
\end{align}
where $\mat \Xi = \diag[{\Vert\vect v_1\Vert_2}^2, {\Vert\vect v_2\Vert_2}^2, ..., {\Vert\vect v_n\Vert_2}^2]$. The leader-follower objective is achieved when there holds $\vect{u} = \vect{v} = \vect{0}_{np}$.

We now present the main result for this Section.

\begin{theorem}\label{theorem_2}
Suppose that each follower agent with dynamics \eqref{eq:el_dynamics_2}, under Assumption~\ref{assm:no_gravity}, employs the controller \eqref{eq:control_input_1} with trigger function \eqref{eq:trigger_function_1}. Suppose further that the directed graph $\mathcal{G}$ contains a directed spanning tree, with the leader agent $0$ as the root node (thus with no incoming edges) and  the follower graph $\mathcal{G}_F$ is undirected. Then the leader-follower consensus objective is globally asymptotically achieved and no agent will exhibit Zeno behaviour.
\end{theorem}

\begin{proof}
We divide our proof into two parts. In the first part, we focus on the stability analysis of the system \eqref{eq:compact_system_dynamic_1}. Notice that \eqref{eq:compact_system_dynamic_1} is non-autonomous in the sense that it is not self-contained ($\mat{M}_i, \mat{C}_i$ depend on $\vect{q}_i$ and $\vect{q}_i, \vect{\dot{q}}_i$ respectively). However, study of a Lyapunov-like function shows leader-follower consensus is achieved. In the second part, analysis is provided to show the exclusion of Zeno behaviour \emph{for each agent}.

\subsubsection{Stability analysis}\label{controller1:stability}
Consider the following Lyapunov-like function
\begin{align}\label{eq:lyapunov_candidate_undirected}
V & =  \frac{1}{2}\vect{u}^\top(\mathcal{\mat L}_{22}\otimes\mat{I}_p)\vect{u} + \frac{1}{2}\vect{v}^\top\mat{M}\vect{v} +\sum_{i=1}^n\frac{1}{2\alpha} (\mu_i-\bar{\mu})^2 \notag\\
  &= V_1 + V_2 +V_3
\end{align}
where $\bar{\mu}$ is a strictly positive constant. The chosen of $\bar{\mu}$ will be presented below. Since $\mathcal{G}$ contains a directed spanning tree and $\mathcal{G}_F$ is undirected, according to \textbf{Lemma} \ref{lem:lem5}, $\mathcal{\mat L}_{22}$ is positive definite. Note that $\mat{M}$ is positive definite and $V_3$ is non-negative, we conclude that $V$ is strictly positive for nonzero $\vect{u}$ and $\vect{v}$.

Taking the derivative of $V$ with respect to time, along the trajectory of system \eqref{eq:compact_system_dynamic_1}, there holds $\dot{V} = \dot{V}_1 + \dot{V}_2 + \dot{V}_3$. Evaluating $\dot{V}_1$ yields
\begin{align*}
\dot{V}_1 = \vect{u}^\top(\mat{L}_{22}\otimes\mat{I}_p)\vect{v}
\end{align*}
Next, the derivative $\dot{V}_2$ is evaluated to be
\begin{align*}
\dot{V}_2 &= \vect{v}^\top\mat{M}\dot{\vect{v}} + \frac{1}{2}\vect{v}^\top\dot{\mat{M}}\vect{v} \notag\\
		  &= -\vect{v}^\top\mat{C}\vect{v} - \vect{v}^\top(\mathcal{\mat L}_{22}\otimes\mat{I}_p)\vect{u} - \vect{v}^\top\mat{K}\vect{v} - \vect{v}^\top\vect{e} + \frac{1}{2}\vect{v}^\top\dot{\mat{M}}\vect{v} \notag\\
          &= - \vect{v}^\top(\mathcal{\mat L}_{22}\otimes\mat{I}_p)\vect{u} - \vect{v}^\top\mat{K}\vect{v} - \vect{v}^\top\vect{e}
\end{align*}
Lastly, $\dot{V}_3$ evaluates to
\begin{align*}
\dot{V}_3 = \sum_{i=1}^n(\mu_i-\bar\mu)\vect{v}_i^\top\vect{v}_i = \vect{v}^\top\mat{K}\vect{v} - \bar{\mu}\vect{v}^\top\vect{v}		  	
\end{align*}
Since $\mat{L}_{22}$ is symmetric, summing $\dot{V}_1$, $\dot{V}_2$ and $\dot{V}_3$ yields
\begin{align}\label{eq:Vdot}
\dot{V} = -\bar{\mu}\vect{v}^\top\vect{v} + \vect{v}^\top\vect{e}	
\end{align}
By using the inequality $\vect{v}^\top\vect{e}\leq\frac{a}{2}\|\vect v\|^2+\frac{1}{2a}\|e\|^2, \forall a>0$, we obtain
\begin{align*}
\dot{V}\leq (\frac{a}{2}-\bar\mu)\|\vect v\|^2 + \frac{1}{2a}\|\vect e\|^2
\end{align*}
Note that the nonpositivity of $f_i(\vect{e}_i,\vect{v}_i,\omega_i)$ guarantees that $\|\vect e\|^2\leq\beta\|\vect v\|^2+\bar{\omega}(t)$, where $\beta=\max_i\{\beta_i\}$ and $\bar{\omega}(t)=\sum_{i=1}^N\omega_i(t)$. It follows that $\dot V$ satisfies
\begin{align*}
\dot{V}\leq (\frac{a}{2}+\frac{\beta}{2a}-\bar\mu)\|\vect v\|^2 + \bar{\omega}(t)
\end{align*}
For notation simplicity, we define $\chi=\bar\mu-\frac{a}{2}-\frac{\beta}{2a}$. Note that for any given $a$ and $\beta$, we can find a sufficiently large $\bar\mu$ to ensure $\chi>0$ and thus
\begin{align*}
\dot V\leq-\chi\|\vect v\|^2+\bar{\omega}(t)
\end{align*}
and it is straightforward to conclude that the parameter $\beta_i$ in the trigger function \eqref{eq:trigger_function_1} can be selected as an arbitrarily positive constant. Integrating both sides of the above equation from 0 to $t$, for any $t>0$, yields
\begin{align*}
V(t)+\chi\int_0^t\|\vect v(\epsilon)\|^2 .\text{d}\epsilon\leq V(0)+\sum_{i=1}^n\frac{\kappa_i}{\varepsilon_i}
\end{align*}
which implies that $V(t)$ and $\chi\int_0^t\|\vect v(\epsilon)\|^2 .\text{d}\epsilon$ are bounded since $V(0), \kappa_i, \varepsilon_i$ are all bounded. By recalling \eqref{eq:lyapunov_candidate_undirected}, it is straightforward to conclude that $\vect u, \vect v, \mu_i$ are all bounded. Now we turn to $\dot{\vect{v}}_i$. Notice that $\dot{\vect{q}}_0=0$ and \eqref{eq:system_dynamic_1}, we have
\begin{align}\label{eq:dot_vi}
\dot{\vect{v}}_i=-\mat{M}_i(\vect{q}_i)^{-1}[\mat{C}_i(\vect{q}_i,\dot{\vect{q}_i})\dot{\vect{q}_i}+\vect{z}_i(t_k^i)]
\end{align}
Since $\vect{u}, \vect{v}, \mu_i$ are bounded, $\dot{\vect{q}_i}$ and $\vect{z}_i(t_k^i)$ are bounded. Then by recalling properties \ref{prty:EL_M_bound} and \ref{prty:EL_C_bound}, we conclude that $\dot{\vect{v}}$ is bounded. From the fact that both $\vect{v}$ and $\dot{\vect{v}}$ are bounded, we obtain $\vect{v},\dot{\vect{v}}\in \mathcal{L}_\infty$. Moreover, the boundedness of $\chi\int_0^t\|\vect v(\epsilon)\|^2 .\text{d}\epsilon$ indicates $\vect v\in \mathcal{L}_2$. By applying \textbf{Lemma} \ref{lem:l1}, we conclude that $\vect{v}\rightarrow \vect{0}_{np}$ as $t\rightarrow\vect\infty$. From \eqref{eq:control_gain} we observe that $\mu_i$ is strictly monotonically increasing. Combining this with the fact that $\mu_i \geq 0$ is bounded, we conclude that $\mu_i(t),\forall\,i$ tends to a finite constant value as $t\to\infty$.

Now we turn to prove that $\vect{u}\rightarrow \vect{0}_{np}$. Due to the difficulty arising from the term $\omega_i(t)$, and the second order dynamics, the proof is more complex than existing proofs for showing convergence to the consensus objective. Consider firstly $\vect e$ and $\vect K$. By recalling the definitions of $\vect{e}_i$ and the trigger function $f_i$, we observe that $\|\vect e\|^2\leq\beta\|\vect v\|^2+\bar{\omega}(t),\forall\,t$. We concluded above that $\lim_{t\to\infty} \|\vect v\|, \bar{\omega}(t) = 0$ which implies that $\lim_{t\to\infty} \vect e = \vect{0}_{np}$. Recalling the definition of $\mat{K}$ above \eqref{eq:compact_system_dynamic_1}), and the fact that $\mu_i,\forall\,i$ tends to a constant value as $t\to\infty$, we conclude that $\lim_{t\to\infty} \mat{K} = \bar{\mat{K}}$ where $\bar{\mat{K}}$ is some finite constant matrix. Rewrite the second equation of \eqref{eq:compact_system_dynamic_1} as
\begin{align}\label{eq:compact_system_dynamic_r}
\dot{\vect{v}} = \vect{f}(t)+\vect{r}(t)
\end{align}
where \mbox{$\vect{f}(t)=-\mat{M}(\vect{q})^{-1}(\mathcal{\mat L}_{22}\otimes \mat{I}_p)\vect{u}$} and $\vect{r}(t)=-\mat{M}(\vect{q})^{-1}\big[\mat{C}(\vect{q},\vect{v})\vect{v}+\mat{K}\vect{v}+\vect{e}\big]$ are both vector functions. Since $\lim_{t\to\infty} \vect v, \vect e = \vect{0}_{np}$, $\mat{K}$ is finite, and $\mat{M}, \mat{C}$ are bounded according to Properties~\ref{prty:EL_M_bound} and \ref{prty:EL_C_bound}, it is obvious that $\lim_{t\to\infty}\vect{r}(t) = \vect{0}_{np}$. Then by integrating both sides of \eqref{eq:compact_system_dynamic_r} from $t$ to $t+\Delta$, where $\Delta$ is a finite positive constant and $t\geq 0$, we obtain
\begin{align}{\label{eq:integral_equality}}
\vect{v}(t+\Delta)-\vect{v}(t)=\int_{t}^{t+\Delta}\vect{f}(s)\text{d}s
+\int_{t}^{t+\Delta}\vect{r}(s).\text{d}s
\end{align}
This implies that there holds
\begin{align}{\label{eq:integral_inequality_1}}
\left\|\int_{t}^{t+\Delta}\vect{f}(s)\text{d}s\right\|&\leq\|\vect{v}(t+\Delta)-\vect{v}(t)\|
+\left\|\int_{t}^{t+\Delta}\vect{r}(s)\text{d}s\right\|
\end{align}
Consider the term $\Vert\int_{t}^{t+\Delta}\vect{f}(s)\text{d}s\Vert$. By applying \textbf{Theorem} \ref{theorem:mean_value_vector}, we conclude that there holds
\begin{align*}
\left\|\int_{t}^{t+\Delta}\vect{f}(s)\text{d}s\right\|\leq\Delta\|\vect{f}(t+\theta(t))\|
\end{align*}
where $\theta(t)\in(0,\Delta)$. Subtracting $\Delta\|\vect{f}(t)\|$ from the both sides of the above inequality yields
\begin{align*}
\left\|\int_{t}^{t+\Delta}\vect{f}(s)\text{d}s\right\|-\Delta\|\vect{f}(t)\| \leq \Delta\left(\|\vect{f}(t+\theta(t))\|-\|\vect{f}(t)\| \right)
\end{align*}
Considering the above right hand side, we observe that $\Delta\left(\|\vect{f}(t+\theta(t))\|-\|\vect{f}(t)\| \right) \leq \Delta \| \vect{f}(t+\theta(t)) -\vect{f}(t) \| = \Delta \| \int_{t}^{t+\theta(t)}\dot{\vect{f}}(s)\text{d}s \|$, which implies that
\begin{align}\label{eq:integral_inequality_2}
\left\|\int_{t}^{t+\Delta}\vect{f}(s)\text{d}s \right\|-\Delta\|\vect{f}(t)\| \leq \Delta \left \| \int_{t}^{t+\theta(t)}\dot{\vect{f}}(s)\text{d}s \right \|
\end{align}
Note that $\text{d}(\mat{M}^{-1})/\text{d}t = -\mat{M}^{-1}\dot{\mat{M}}\mat{M}^{-1}$ because $\text{d}(\mat{M}^{-1} \mat{M})/\text{d}t = \mat{M}^{-1}\dot{\mat{M}} + (\text{d}(\mat{M}^{-1})/\text{d}t)\mat{M} = \mat{0}$. From Properties~\ref{prty:EL_C_bound} and \ref{prty:EL_CM_skew}, we observe that $\lim_{t\to\infty} \| \dot{\mat{M}} \| \leq 2k_C \| \vect{v} \| = 0$. Observe that
\begin{align*}
\dot{\vect{f}}=-\left(\frac{\text{d}(\mat{M}(\vect{q})^{-1})}{\text{d}t}(\mathcal{\mat L}_{22}\otimes \mat{I}_p)\vect{v}+\mat{M}(\vect{q})^{-1}(\mathcal{\mat L}_{22}\otimes \mat{I}_p)\vect{v}\right)
\end{align*}
We proved below \eqref{eq:dot_vi} that $\vect{u}$ is bounded and $\lim_{t\to\infty}\vect{v} = \vect{0}_{np}$. Recall also that $\|\mat{M}(\vect{q})^{-1}\|$ is bounded according to Property~\ref{prty:EL_M_bound}. It follows that $\lim_{t\to\infty} \|\dot{\vect{f}} \| = 0$ which implies that $\| \int_{t}^{t+\theta(t)}\dot{\vect{f}}(s)\text{d}s \| = 0$ since $\theta(t) \in (0,\Delta)$. The inequality \eqref{eq:integral_inequality_2} then implies that $\lim_{t\to\infty} \left\|\int_{t}^{t+\Delta}\vect{f}(s)\text{d}s\right\| = \Delta\|\vect{f}(t)\|$. By substituting this into the left hand side of \eqref{eq:integral_inequality_1}, we obtain
\begin{align}\label{eq:inequality_01}
\Delta\|\vect{f}(t)\|&\leq\|\vect{v}(t+\Delta)-\vect{v}(t)\|
+\left\|\int_{t}^{t+\Delta}\vect{r}(s)\text{d}s\right\|
\end{align}
as $t\rightarrow\infty$. Immediately above \eqref{eq:integral_equality}, we showed that $\lim_{t\to\infty} \vect{r}=\vect{0}_{np}$. In addition, $\lim_{t\to\infty} \vect{v} = \vect{0}_{np}$ and $\Delta$ is a positive constant. We conclude that $\lim_{t\to\infty} \|\vect{v}(t+\Delta)-\vect{v}(t)\|+\left\|\int_{t}^{t+\Delta}\vect{r}(s)\text{d}s\right\| = 0$, which according to \eqref{eq:inequality_01} implies that $\lim_{t\to\infty} \|\vect{f}(t)\| = 0$. By recalling that $\vect{f}(t)=-\mat{M}(\vect{q})^{-1}(\mathcal{\mat L}_{22}\otimes\mat{I}_p)\vect{u}$, we conclude $\lim_{t\to\infty} \vect{u} = \vect{0}_{np}$ since both $\mat{M}(\vect{q})^{-1}$ and $\mat{L}_{22}$ are both positive definite. It is obvious that $\lim_{t\to\infty} \vect{u}, \vect{v} = \vect{0}_{np}$ implies the leader-follower objective is asymptotically achieved.

\subsubsection{Absence of Zeno behaviour}
According to \textbf{Definition} \ref{def:1}, we can prove that Zeno behaviour does not occur for $t\in[0, b]$ by showing that for all $k \geq 0$ there holds $t_{k+1}^i - t_k^i \geq \xi$ where $\xi > 0$ is a strictly positive constant.

Let $\xi_i$ denote the lower bound of the inter-event interval $t_{k+1}^{i}-t_k^i$ for agent $i$, i.e. $t_{k+1}^{i}-t_k^i\geq \xi_i~,\forall k : t_k^i \in [0,b]$. In this part of the proof, we show that $\xi_i$ is strictly positive for $k<\infty$ and thus no Zeno behaviour can occur. From the definition of $\vect{e}_i(t)$ in \eqref{eq:measurement_error} and the fact that $\vect{z}_i(t_k^i)$ is a constant, we observe that the derivative of $\|\vect{e}_i(t)\|$ with respect to time satisfies
\begin{align}
\frac{d}{dt}\|\vect{e}_i(t)\|\leq\|\dot{\vect{z}}_i(t)\|    \label{eq:r10}
\end{align}
where $
\dot{\vect z}_i(t)=\sum_{j=0}^{n}a_{ij}(\dot{\vect q}_i(t)-\dot{\vect q}_j(t))+\dot{\mu}_i(t)\vect{\dot{q}}_i(t)+\mu_i(t)\vect{\ddot{q}}_i(t),~i=1,\ldots,n
$. Note that it is straightward to conclude $\dot{\vect q}_i(t), \ddot{\vect q}_i(t), \dot{\mu}_i(t), \mu_i(t)$ are bounded according to the arguments in \emph{Part} 1). This implies $\dot{\vect z}_i(t)$ is bounded. By letting a positive constant $B_e$ represent the upper bound of $\|\dot{\vect z}_i(t)\|$, we obtain
\begin{align*}
\frac{d}{dt}\|\vect{e}_i(t)\|\leq B_e
\end{align*}
It follows that
\begin{align}
\|\vect{e}_i(t)\|\leq \int_{t_k^i}^tB_edt=B_e(t-t_k^i)   \label{eq:r13}
\end{align}
for $t\in[t_k^i, t_{k+1}^i)$ and for any $k$. It is obvious that the next event time $t_{k+1}^{i}$ is determined both by the changing rate of $\|\vect{e}_i(t)\|$ and by the value of the comparison term $\beta_i\|\vect{z}_i(t)\|^2+\mu_i(t)$ at $t_{k+1}^{i}$. Moreover, $t_{k+1}^{i}$ is the time that
\begin{align}\label{eq:threshold}
\|\vect{e}_i(t)\|^2=\beta_i\|\vect{v}_i(t)\|^2+\omega_i(t),~~~~t>t_k^i
\end{align}
holds. In \emph{Part} 1) we conclude that global state variable $\vect{v}(t) \to \vect{0}_{np}$ as $t\to\infty$ but notice that in the evolution of the system \eqref{eq:compact_system_dynamic_1}, the state variable $\vect{v}_{i}(t)$ may be equal to $\vect{0}_p$ instantaneously ($\vect{v}_{i}(t)$ is a component of $\vect{v}(t)$) may also hold at $t_{k+1}^i$. However, this does not imply leader-follower consensus is reached since $\dot{\vect v}_{i}(t)$ can be non-zero at $t_{k+1}^i$. We refer to such points as ``zero-crossing points'' for convenience. Here we provide Fig. \ref{fig:one_arm} to show the trigger performance at the zero-crossing points of $\vect{v}_{i}(t)$ when $\omega_i(t)=0$. \emph{It is observed that dense trigger behaviour occurs whenever $\vect{v}_{i}(t)$ crosses zero. Theoretically, it can be proved that Zeno behaviour takes place at those zero-crossing points. We refer interested readers to \cite{sun2016new} with detail arguments of the Zeno triggering issues at zero-crossing points.} 

Now we return to the trigger time interval analysis. By recalling \eqref{eq:threshold}, we conclude that at $t_{k+1}^i$, the triggering of the event can only occur according to the following two cases:
\begin{itemize}
\item
Case 1: If $\|\vect{v}_i(t_{k+1}^i)\|\neq0$, the equality $\|\vect{e}_i(t_{k+1}^i)\|=\beta_i\|\vect{v}_i(t_{k+1}^i)\|^2+\omega_i(t_{k+1}^i)$ is satisfied.
\item
Case 2: If $\|\vect{v}_i(t_{k+1}^i)\|=0$, the equality $\|\vect{e}_i(t_{k+1}^i)\|=\omega_i(t_{k+1}^i)$ is satisfied.
\end{itemize}
Compare the above two cases, and note that $\| \vect{v}_i(t^i_{k+1})\| > 0$ for any $\|\vect{v}_{i}(t_{k+1}^i)\|\neq 0$. By recalling that $e_i(t)$ is equal to zero at $t_k^i$, it is straightforward to conclude that it takes longer for the quantity $\|e_i(t)\|^2$ to increase to be equal to the quantity $\beta_i\|\vect{v}_i(t_{k+1}^i)\|^2+\omega_i(t_{k+1}^i)$ (i.e. Case 1) than to increase to be equal to the quantity $\omega_{i}(t_{k+1}^i)$ (i.e. Case 2) and thus trigger an event and resetting $\vect{e}_i(t)$. This implies that $\xi_{Case\:2} < \xi_{Case\:1}$ and proving that there exists a strictly positive $\xi_{Case\:2}$ allows us to draw the conclusion that no Zeno behaviour occurs. According to \eqref{eq:r13}, we have
\begin{align*}
B_e\xi_{Case\:2}\geq\omega_{i}(t)=\exp(-\kappa_{i}(t_k^i+\xi_{Case\:2}))
\end{align*}
This implies that the inter-event time $\xi_{Case\:2}$ is lower bounded by the solution of the following equation
\begin{align}\label{eq:inter_event_bound}
B_e\xi_{Case\:2}=\exp(-\kappa_{i}(t_k^i+\xi_{Case\:2}))
\end{align}
The solution is time-dependent and strictly positive for any finite time since $B_e$ is strictly positive and upper bounded. Zeno behaviour is thus excluded for all agents.
\end{proof}

\begin{rek}
The reader will have noticed the complexity and length of argument required to go from concluding $\lim_{t\to\infty} \vect v = \vect 0_{np}$ below \eqref{eq:dot_vi}, to concluding $\lim_{t\to\infty} \vect{u} = \vect 0_{np}$ below \eqref{eq:inequality_01}. The key reason is the combination of second-order dynamics and the non-autonomous nature of the networked system \eqref{eq:compact_system_dynamic_1} resulting from the offset term $\omega_i(t)$ in the trigger function \eqref{eq:trigger_function_1}. The authors in \cite{huang2016EL_eventtrigger} use a similar trigger function with the same offset term, and claim that $\lim_{t\to\infty} \vect v = \vect 0_{np}$ implies that $\lim_{t\to\infty} \dot{\vect v} = \vect 0_{np}$. This is not correct since the system is non-autonomous. The paper \cite{li2015eventtrigger} uses a trigger function without the offset term, and thus they are able to avoid the non-autonomous issue. However, the lack of the offset term can yield Zeno behaviour, something which was not recorded by \cite{li2015eventtrigger}. We explore the use of the offset term for avoiding Zeno behaviour in the next section.
\end{rek}

\begin{rek}
Unfortunately, we cannot find a constant lower bound for the inter-event time interval. The lower bound $\xi_{Case2}$ found by solving \eqref{eq:inter_event_bound} is still time-dependent and tends to zero as $t\rightarrow\infty$. The avoidance of Zeno behaviour depends on the exponential decay offset completely and the trigger performance when $t\rightarrow \infty$ is not discussed in the theoretical analysis.  However, we note that the state-dependent term in \eqref{eq:trigger_function_1} provides a performance advantage when $t\rightarrow\infty$ due to its own specific effects and should not be removed. We will provide detail explanations for the advantages of our proposed trigger function \eqref{eq:trigger_function_1} in the following subsection. 
\end{rek}

\subsection{Discussions on the chosen of trigger functions} \label{sec3-b}
In this subsection, we provide discussions regarding the trigger performance of controller \eqref{eq:control_input_1} under the following three trigger functions
\begin{itemize}
\item State-dependent trigger function (SDTF)
\begin{align}\label{eq:state_trigger_function}
f_i=\|\vect{e}_i(t)\|-\beta_i\|\vect{v}_i(t)\| 
\end{align}
\item Time-dependent trigger function (TDTF)
\begin{align}\label{eq:time_trigger_function}
f_i=\|\vect{e}_i(t)\|-\kappa_i\exp(-\varepsilon_it) 
\end{align}
\item Mixed trigger function (MTF), which is the proposed \eqref{eq:trigger_function_1}
\begin{align} \label{eq:mixed_trigger_function}
f_i=\|\vect{e}_i(t)\|^2-\beta_i\|\vect{v}_i(t)\|^2-\kappa_i\exp(-\varepsilon_it)
\end{align}
\end{itemize}
from both the viewpoints of theoretical analysis and numerical simulations. In doing so, we highlight the advantages
of our proposed trigger function \eqref{eq:trigger_function_1}. Note that it is hard, but not impossible, to observe the zero-crossing phenomenon for $\vect{v}_i(t) \in\mathbb{R}^p, p \geq 2$ (i.e. when $\vect v_i(t) = \vect 0_p$ occurs, Zeno behaviour is observed as discussed in the proof of Theorem~\ref{theorem_2} and in \cite{sun2016new}). This is because each entry of $\vect v_i(t)$ must be simultaneously equal to $0$. For purposes of illustration, in this subsection we therefore simulate using  dynamics of a one-arm mechanic manipulator ($\vect v_i(t)\in\mathbb{R}^1$). The dynamics are described by equation 3.5 in \cite{kelly2006control}. For all simulations presented in this subsection, we set a constant step size in MATLAB to be 0.00005 seconds (the numerical accuracy of the simulation) and the running time to be 30 seconds. In order to compare performance, we require the following two definitions
\begin{definition}[Minimum Inter-Event Time for Agent $i$]\label{def:min_IE_time}
For $j = \{\text{SDTF,TDTF,MTF}\}$, and for $i = \{1, \hdots, n\}$ define the minimum inter-event time for Agent $i$ $\Delta_{j}^i$ as $\Delta_j^i \triangleq \min_{k} t^i_{k+1} - t^i_k$.
\end{definition}

\begin{definition}[Infimum Time of $\Delta_j$ For Agent $i$]\label{def:inf_time_delta}
For $j = \{\text{SDTF,TDTF,MTF}\}$, and for $i = \{1, \hdots, n\}$, define the ``infimum time of $\Delta_j$ for Agent $i$'' as $t_{\Delta_j^i} \triangleq \inf_{t_k^i} t_k^i : t^{k+2} - t_{k+1} = t^{k+1}_i - t^k_i = \Delta_j^i,\forall k$.
\end{definition}
In other words, for Agent $i$, $t_{\Delta_j^i}$ is the infimum of all event times $t_k^i, \forall\,i$ such that the inter-event time between consecutive events $k+1$ and $k+2$ is equal to the minimum inter-event time $\Delta_j^i$. If there are multiple consecutive events with inter-event time $\Delta_j^i > 0$ then we call this a \emph{dense triggering of events}. Note that because $\Delta_j^i > 0$, dense triggering is not Zeno behaviour, but is nevertheless undesirable.

\begin{table*}[!ht]
\subfigure{
\begin{minipage}{0.5\linewidth}
\begin{center}
\includegraphics[width=0.7\linewidth]{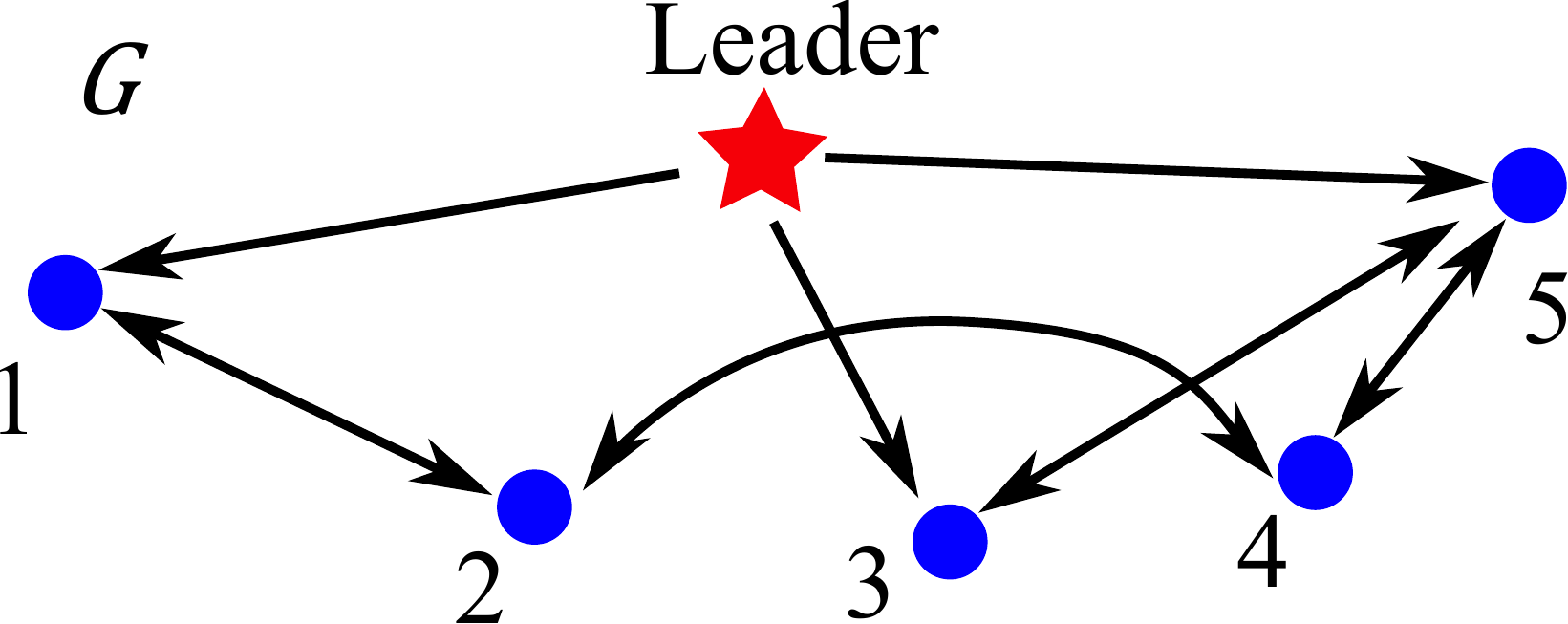}
\captionof{figure}{Graph topology used in simulations}  
\label{fig:topology}
\end{center}
\end{minipage}
\hfill
}
\subfigure{
\begin{minipage}{0.5\linewidth}
\caption{Number of events for three different trigger functions}
\begin{center}
{\renewcommand{\arraystretch}{1.1}
\begin{tabular}{c c c c c}
• & \text{State-dependent} & \text{Time-dependent} & \text{Mixed}  \\\hline
Agent 1 & $259$ & $5581$ & 74     \\\hline
Agent 2 & $184$ & $8106$ & 63      \\\hline
Agent 3 & $575$ & $2251$ & 168      \\\hline
Agent 4 & $94$ & $7365$ & 101       \\\hline
Agent 5 & $438$ & $3845$ & 200     \\\hline
Total   & $1550$ & $27148$ & 606   \\\hline
\end{tabular}
}\label{ta:event_numbers}
\end{center}
\end{minipage}
\hfill
}
\subfigure{
\begin{minipage}{0.5\linewidth}
\caption{Minimum inter-event time $\Delta_j^i$ under three trigger functions}
\begin{center}
{\renewcommand{\arraystretch}{1.1}
\begin{tabular}{c c c c c}
• & \text{State-dependent} & \text{Time-dependent} & \text{Mixed}  \\\hline
Agent 1 & $0.00005$ & $0.00005$ & 0.0388    \\\hline
Agent 2 & $0.00005$ & $0.00005$ & 0.0235    \\\hline
Agent 3 & $0.00005$ & $0.00005$ & 0.0010    \\\hline
Agent 4 & $0.00005$ & $0.00005$ & 0.0037    \\\hline
Agent 5 & $0.00005$ & $0.00005$ & 0.0006    \\\hline
\end{tabular}
}
\end{center}
\label{ta:event_interval}
\end{minipage}
\hfill
}
\subfigure{
\begin{minipage}{0.5\linewidth}
\caption{The infimum Time of $\Delta_j$, $t_{\Delta_j^i}$ as defined in Definition~\ref{def:inf_time_delta}}
\begin{center}
{\renewcommand{\arraystretch}{1.1}
\begin{tabular}{c c c c c}
• & \text{State-dependent} & \text{Time-dependent} & \text{Mixed}  \\\hline
Agent 1 & $0.6736$ & $29.8177$ & 18.0246     \\\hline
Agent 2 & $0.3042$ & $29.6469$ & 2.3991    \\\hline
Agent 3 & $0.4722$ & $29.8398$ & 14.0583     \\\hline
Agent 4 & $1.3219$ & $29.6458$ & 1.3182     \\\hline
Agent 5 & $0.0798$ & $29.9830$ & 15.435    \\\hline
\end{tabular}
}\label{ta:event_time}
\end{center}
\end{minipage}
\hfill
}
\subfigure{
\begin{minipage}{0.5\linewidth}
\begin{center}
\includegraphics[width=\linewidth]{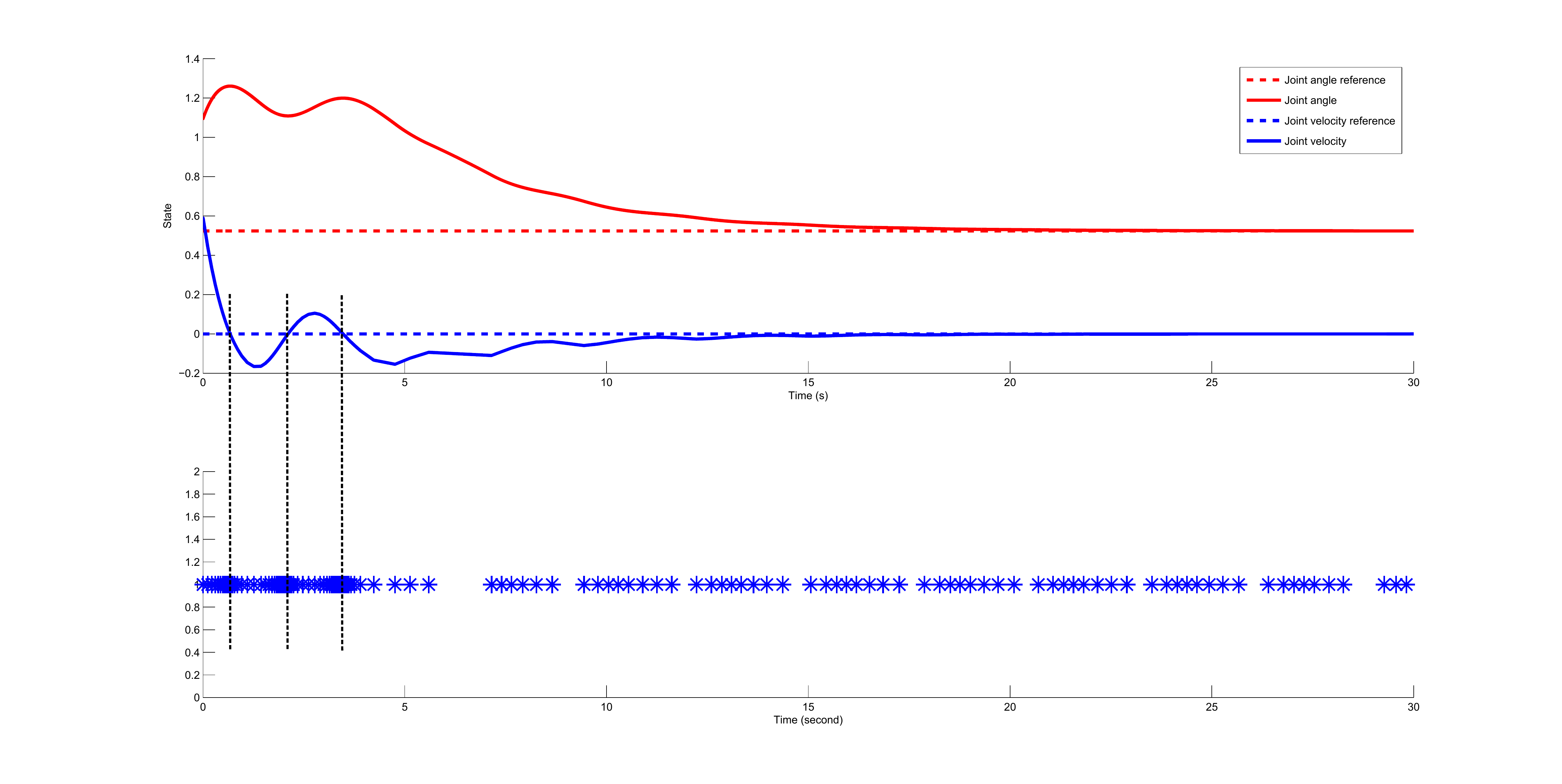}
\captionof{figure}{Top: the evolutions of the generalized coordinate and velocity of agent 1. Bottom: the trigger event times of agent 1} 
\label{fig:one_arm}
\end{center}
\end{minipage}
}
\subfigure{
\begin{minipage}{0.5\linewidth}
\begin{center}
\includegraphics[width=\linewidth]{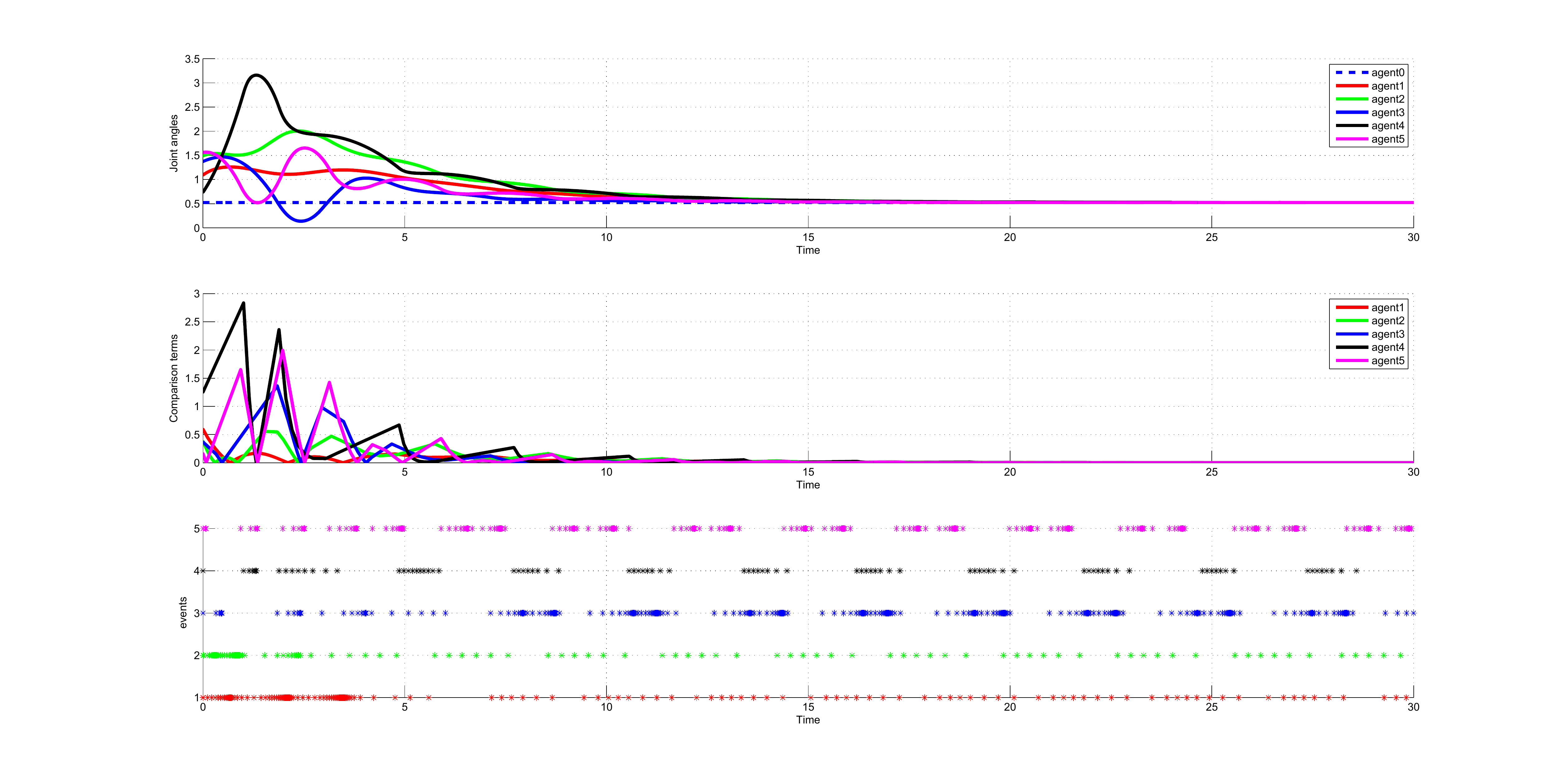}
\captionof{figure}{Performance of controller \eqref{eq:control_input_1} using  SDTF \eqref{eq:state_trigger_function}. We set $\beta_i=2.4$. From top to bottom: 1) the rendezvous of the generalized coordinates; 2) the evolution of $\beta_i\|v_i(t)\|$; 3) event times for each agent.} 
\label{fig:state_trigger_function}
\end{center}
\end{minipage}
}
\subfigure{
\begin{minipage}{0.5\linewidth}
\begin{center}
\includegraphics[width=\linewidth]{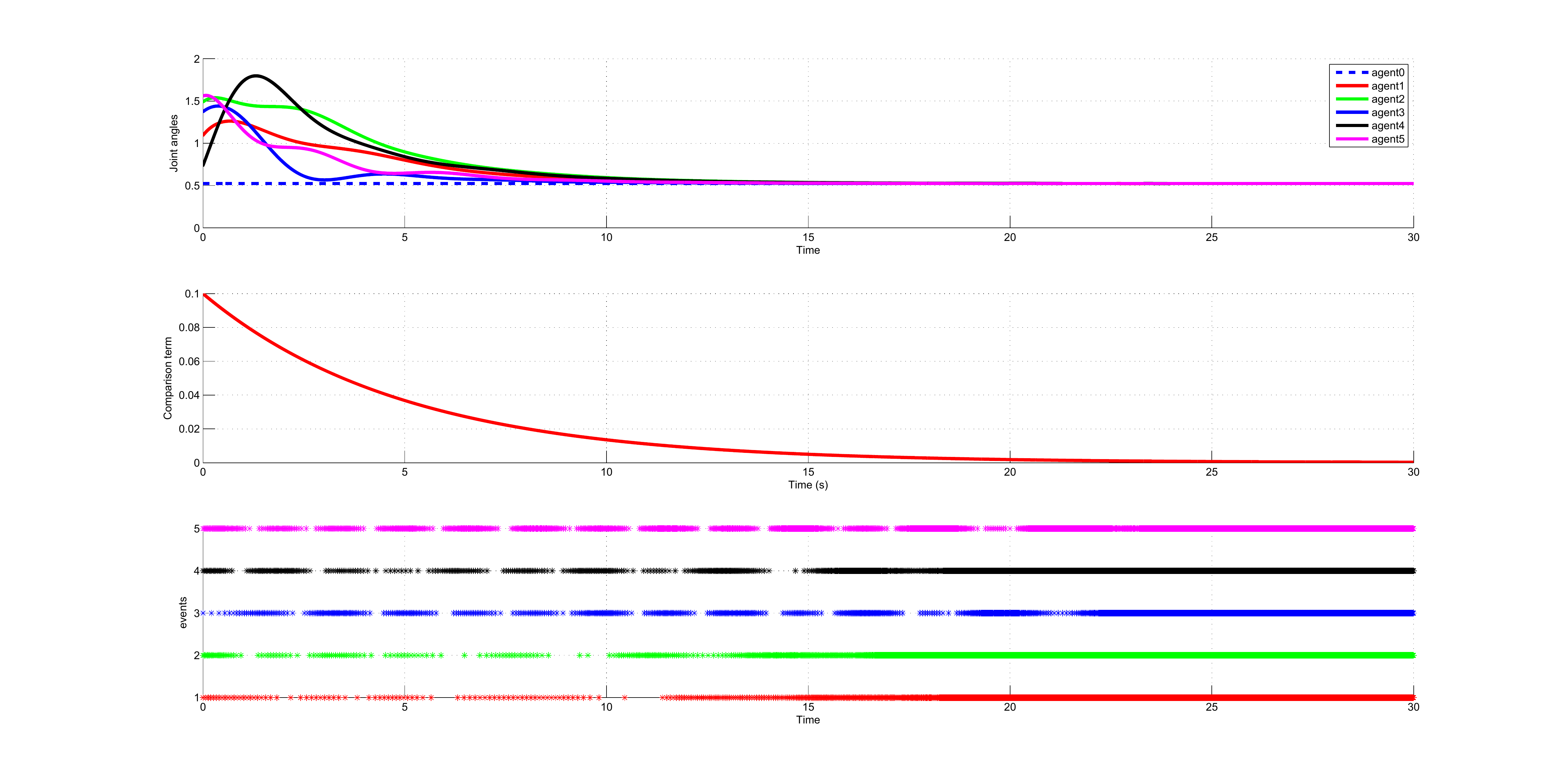}
\captionof{figure}{Performance of controller \eqref{eq:control_input_1} using TDTF \eqref{eq:time_trigger_function}. We set $\kappa_i\exp(-\varepsilon_it) = 0.1\exp(-0.2t)$. From top to bottom: 1) the rendezvous of the generalized coordinates; 2) the evolution of $\kappa_i\exp(-\varepsilon_it)$; 3) event times for each agent.}  
\label{fig:time_trigger_function}
\end{center}
\end{minipage}
}
\subfigure{
\begin{minipage}{0.5\linewidth}
\begin{center}
\includegraphics[width=\linewidth]{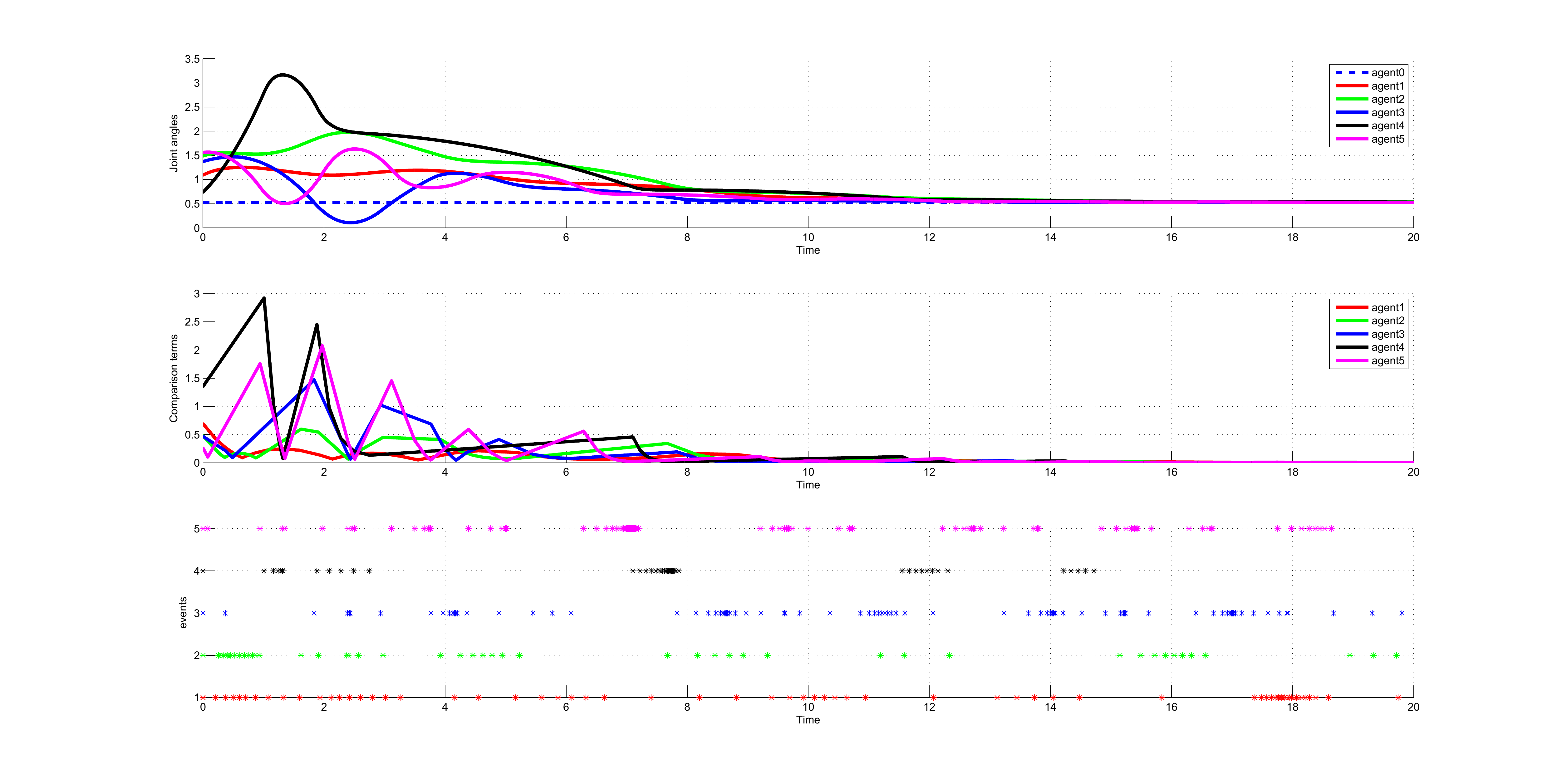}
\captionof{figure}{Performance of controller \eqref{eq:control_input_1} with MTF \eqref{eq:mixed_trigger_function}. We set $\beta_i= 2.4$ and $\kappa_i\exp(-\varepsilon_it) = 0.1\exp(-0.2t)$. From top to bottom: 1) the rendezvous of the generalized coordinates; 2) the evolution of $\beta_i\|v_i(t)\|+\kappa_i\exp(-\varepsilon_it)$; 3) event times for each agent.} 
\label{fig:mixed_trigger_function}
\end{center}
\end{minipage}
}
\end{table*}

Due to space limitations and the similarity of the proofs, we omit the proofs of convergence of system \eqref{eq:system_dynamic_1} under trigger functions \eqref{eq:state_trigger_function} and \eqref{eq:time_trigger_function}. Figures \ref{fig:state_trigger_function} and \ref{fig:time_trigger_function} illustrate the controller \eqref{eq:control_input_1} using SDTF \eqref{eq:state_trigger_function}, and TDTF \eqref{eq:time_trigger_function}, respectively. The figures show  rendezvous of the generalized coordinates, the evolutions of comparison terms ($\beta_i\|\vect v_i(t)\|$ in SDTF and $\kappa_i\exp(-\varepsilon_it)$ in TDTF) and event times. Figure \ref{fig:mixed_trigger_function} shows the performance of controller \eqref{eq:control_input_1} using MTF. We also provide three tables to compare the trigger performance when using SDTF, TDTF and MTF. Table \ref{ta:event_numbers} records the total number of events which occur when using the three different trigger functions. Table \ref{ta:event_interval} records the minimum inter-event time, $\Delta_j^i$. Table \ref{ta:event_time} records the infimum time value, $t_{\Delta_j^i}$, which was defined in Definition~\ref{def:inf_time_delta} above.


Note that the SDTF and TDTF are widely adopted in event-based multi-agent consensus literature. We hereby review and illustrate the advantages and disadvantages regarding the trigger performance by separately using SDTF and TDTF.

\subsubsection{SDTF}
The papers \cite{dimarogonas2012distributed,fan2013distributed,garcia2013decentralised,hu2016consensus, liu2016decentralised} used SDTF to determine the event times. The disadvantage of adopting state-dependent trigger function is that Zeno behaviour can occur when the local state-dependent term crosses zero at a finite time value as indicated in \cite{sun2016new} (i.e. in \eqref{eq:state_trigger_function}, the term $v_i(t) = 0$ instantaneously, for $t<\infty$). According to the first column of Table \ref{ta:event_interval}, the minimum inter-event time is $\Delta_{\text{SDTF}}^i = 0.00005$ seconds, for all $i$, which is equal to the fixed time step of the MATLAB simulations. As shown in \cite{sun2016new}, Zeno behaviour will occur at these instants. From the first column of Table \ref{ta:event_time} and the second sub-graph of Fig. \ref{fig:state_trigger_function}, we observe that Zeno behaviour occurs at the time instants that $v_i(t)$ crosses $0$, which supports the conclusion of Zeno behaviour. However, according to the arguments in \cite{dimarogonas2012distributed, sun2016new}, if each agent uses SDTF, then at any time $t$, there exists at least one agent for which the next inter-event interval is strictly positive at any time $t$. In other words, for all $t \in [0,\infty)$, some agents may exhibit Zeno behaviour, \emph{but at least one agent will be Zeno free}. 


\subsubsection{TDTF}
In \cite{seyboth2013event, yang2016decentralized, wei2017edge}, by using carefully-designed TDTF (typically the decay rate of $\exp(-\epsilon_i t)$ in \eqref{eq:time_trigger_function} must be upper bounded), a strictly positive and \emph{constant lower bound} on the inter-event time interval for each agent can be obtained. However, the use of the TDTF has the following two limitations: 1) the applied system has to be exponentially stable, and 2) accurate information (agent's dynamic model and network topology) is required to design the decay rate of $\exp(-\epsilon_i t)$. We emphasise that the use of TDTF with arbitrary decay rate for $\exp(-\epsilon_i t)$ in enough to exclude Zeno behaviour (see the second part of the proof of \textbf{Theorem} \ref{theorem_2}). However, if the decay rate is not selected to be sufficiently slow, the lower bound on the inter-event time cannot be guaranteed to be constant, but instead becomes time dependent. This results in dense triggering behaviour as consensus is almost reached, i.e. multiple events occur in a very short time interval (see Fig. \ref{fig:time_trigger_function}). From the second columns of Table \ref{ta:event_interval} and Table \ref{ta:event_time}, it is observed that $\Delta_j^i$ occurs around $29s$, for all $i$, which is when the system is close to consensus. Note that dense triggering as $t\rightarrow\infty$ is not Zeno behaviour (See Definition \ref{def:1}). However, it can be observed from Table \ref{ta:event_numbers} that unsuitably chosen trigger function will introduce a large amount of events, which is obviously undesirable.    

\subsubsection{MTF}
According to Table \ref{ta:event_numbers}, it is straightforward to conclude that using a MTF shows the best trigger performance with the least number of total events. According to Table \ref{ta:event_numbers}, using MTF also shows that the minimum inter-event time, $\Delta_j^i$ is greater than the constant MATLAB step size of $0.00005$ seconds, which indicates Zeno behaviour is excluded. These observations reveal that MTF is able to combine the advantages of using SDTF and TDTF separately, i.e., the exclusion of Zeno behaviour in finite time (TDTF) and guarantee that dense triggering does not occur as consensus is reached (SDTF). However, a thorough analysis to find a \emph{constant} lower bound on $\Delta_j^i$ when using MTF remains an open challenge (we can find a time-dependent bound).

\begin{rek}
The work \cite{zhu2015event, wang2017event} also use MTF. However, the authors design the evolution speeds of the exponential functions using exact knowledge of agent dynamic models and the graph topology. The effects of adding state-dependent terms to the trigger functions were not well-addressed by the authors of \cite{zhu2015event, wang2017event}. 
\end{rek}

\section{Model-Independent Controller on Directed Graph}\label{sec:dir_IC}
In this section, we propose and analyse a distributed event-triggered model-independent algorithm to achieve leader-follower consensus on a directed network where each fully-actuated agent has self-dynamics described by the Euler-Lagrange equation. For design of the control laws, the following assumption is required.

\begin{assumption}[Limited Use of Centralised Design]\label{assm:central_design}
Three parameters in the algorithm in this Section must be designed to exceed several lower bounding inequalities. These inequalities require knowledge of the constants $k_{\underline{m}}, k_{\overline{M}}, k_C$ defined in the properties \ref{prty:EL_M_bound} and \ref{prty:EL_C_bound}. We therefore assume these constants are known to the designer. 
\end{assumption}

\subsection{Main Result}
Let the triggering time sequence of agent $i$ be $t_{0}^{i}, t_{1}^{i}, \ldots, t_{k}^{i},\ldots$ with \mbox{$t_0^i := 0$}. Consider a model-independent, event-triggered algorithm for the $i^{th}$ follower agent of the form
\begin{align}\label{eq:agent_control_law_direct}
\vect{\tau}_i(t) &= -\sum_{j \in \mathcal{N}_i}  a_{ij}\Big((\vect{q}_i(t_k^i)-\vect{q}_j(t_k^i))+\mu(\dot{\vect{q}}_i(t_k^i)-\dot{\vect{q}}_j(t_k^i)) \Big)\notag\\
&~~~~~~~~~~~~~~~~~~~~~~~~~~~~~~~t\in[t_k^i,t_{k+1}^i)
\end{align}
where $a_{ij}$ is the weighted $(i, j)$ entry of the adjacency matrix $\mathcal{A}$ associated with the weighted directed graph $\mathcal{G}$. The control gain scalar $\mu > 0$ is universal to all agents. To ensure the control objective is achieved, $\mu$ must be designed to satisfy several inequalities, which will be detailed below. Note that if the leader is a neighbour of agent $i$ then for $j= 0$ we have $\mu(\dot{\vect{q}}_i(t_k^i)-\dot{\vect{q}}_0(t_k^i)) = \mu(\dot{\vect{q}}_i(t_k^i))$, which is simply a damping term.

Define a new variable
\begin{align*}
\vect{z}_i(t)=\sum_{j\in\mathcal{N}_i}a_{ij}\Big((\vect{q}_i(t)-\vect{q}_j(t))+\mu(\dot{\vect{q}}_i(t)-\dot{\vect{q}}_j(t)) \Big)  
\end{align*}
We define a state mismatch for agent $i$ between consecutive event times $t_k^i$ and $t_{k+1}^i$ as follows:
\begin{align}
\vect{e}_i(t)&=\vect{z}_i(t_k^i)-\vect{z}_i(t)  \label{eq:measurement error}
\end{align}
for $t\in[t_k^i,t_{k+1}^i)$.

The trigger function is proposed as follows:
\begin{align}\label{eq:agent_trigger_direct}
f_{i}(e_{i}(t)) & =\|\vect e_i(t)\|^2- \mu^{-2}{\beta_1}^2  \Vert \sum_{j\in\mathcal{N}_i}a_{ij}(\vect{q}_i(t)-\vect{q}_j(t)) \Vert^2 \nonumber \\
& \; - {\beta_2}^2  \Vert \sum_{j\in\mathcal{N}_i}a_{ij}(\vect{v}_i(t)-\vect{v}_j(t)) \Vert^2 -\omega_{i}(t)
\end{align}
where $\omega_i(t)=a_i\exp(-\kappa_i t)$ where $a_i,\kappa_i>0$. The parameters $\beta_1$ and $\beta_2$ are to be determined in the sequel. The $k^{th}$ event for agent $i$ is triggered as soon as the trigger condition $f_{i}(\vect e_{i}(t))=0$ is fulfilled at $t=t_k^i$. For $t\in[t_k^i, t_{k+1}^i)$, the control input is $\tau_i(t) = \tau_i(t_k^i)$; the control input is updated when the next event is triggered. Furthermore, every time an event is triggered, and in accordance with their definitions, the measurement error $\vect e_i(t)$ is reset to be equal to zero and thus the trigger function assumes a negative value. One can immediately observe that for all $t$ 
\begin{align*}
\Vert \vect e_i(t) \Vert^2 & \leq \mu^{-2}{\beta_1}^2 \Vert \sum_{j\in\mathcal{N}_i}a_{ij}(\vect{q}_i(t)-\vect{q}_j(t)) \Vert^2 \nonumber \\
& \; + {\beta_2}^2  \Vert \sum_{j\in\mathcal{N}_i}a_{ij}(\vect{v}_i(t)+\vect{v}_j(t)) \Vert^2 + \omega_{i}(t)
\end{align*}
and note that $\sum_{j\in\mathcal{N}_i}a_{ij}(\vect{q}_i(t)-\vect{q}_j(t)) = \sum_{j\in\mathcal{N}_i}a_{ij}[(\vect{q}_i(t)-\vect{q}_0)-(\vect{q}_j(t)-\vect{q}_0)] =\vect l_i^\top \vect u$ where $\vect l_i^\top$ is the $i^{th}$ row of $\mathcal{L}_{22}$. Likewise, $\sum_{j\in\mathcal{N}_i}a_{ij}(\vect{v}_i(t)-\vect{v}_j(t)) = \vect l_i^\top \vect v$. The stacked column vector $\vect e = [\vect e_1^\top, \hdots, \vect e_n^\top]^\top$ then has the following property
\begin{align}
\Vert \vect e \Vert^2 & = \sum_{i = 1}^n \Vert \vect e_i(t) \Vert^2 \notag\\
& \leq \sum_{i = 1}^n \Big( \mu^{-2}{\beta_1}^2 \Vert \vect l_i^\top \vect u \Vert^2 + {\beta_2}^2 \Vert \vect l_i^\top \vect v \Vert^2 + \omega_{i}(t) \Big)
 \end{align}
It is straightforward to verify that $\sum_{i = 1}^n \Vert \vect l_i^\top \vect u \Vert^2 = \Vert \mathcal{L}_{22} \vect u \Vert^2$, and $\sum_{i = 1}^n  \Vert \vect l_i^\top \vect v \Vert^2 = \Vert \mathcal{L}_{22} \vect v \Vert^2$. It then follows that 
\begin{align}
\Vert \vect e \Vert^2 & \leq \mu^{-2}{\beta_1}^2 \Vert \mathcal{L}_{22} \vect u \Vert^2 + {\beta_2}^2 \Vert \mathcal{L}_{22} \vect v \Vert^2 + \bar\omega(t)^2 \\
\Vert \vect e \Vert & \leq \mu^{-1}\beta_1 \Vert \mathcal{L}_{22}\Vert \Vert \vect u \Vert + \beta_2 \Vert \mathcal{L}_{22}\Vert \Vert \vect v \Vert + \bar\omega(t) \label{eq:e_norm_bound}
 \end{align}
where $\bar \omega(t) = (\sum_{i=1}^n \omega_i(t))^{\frac{1}{2}}$

It is obvious that
\begin{align*}
\vect \tau_i(t) = \vect{z}_i(t)+\vect{e}_i(t)
\end{align*}
Applying control law \eqref{eq:agent_control_law_direct} to each agent we can express the networked system using the new variables $\vect{u}, \vect{v}$ as below
\begin{equation}\label{eq:euler_lagrange_general_system_error}
\mat{M}(\vect{q})\dot{\vect{v}} + \mat{C}(\vect{q},\vect{v})\vect{v} + (\mathcal{L}_{22} \otimes \mat{I}_p)(\vect{u} + \mu\vect{v}) + \vect e= \vect{0}
\end{equation}
and expressed as the non-autonomous system
\begin{align}\label{eq:euler_Lagrange_nonautonomous_network}
\dot{\vect{u}} & = \vect{v} \notag \\
\dot{\vect{v}} & = -\mat{M}(\vect{q})^{-1}\left[\mat{C}(\vect{q},\vect{v})\vect{v}+(\mathcal{L}_{22} \otimes \mat{I}_p)(\vect{u} + \mu\vect{v}) + \vect e\right] 
\end{align}
By using arguments like those of usual Lyapunov theory, we will be able to prove the stability of \eqref{eq:euler_Lagrange_nonautonomous_network}. Before we present the main theorem of this section, we state a mild assumption used \emph{only in this Section}.

\begin{assumption}\label{assm:IC}
All possible initial conditions lie in some fixed but arbitrarily large set, which is known a priori. In particular, $\Vert \vect{u}_i(0) \Vert \leq k_a/\sqrt{n}$ and $\Vert \vect{v}_i(0) \Vert_2 \leq k_b/\sqrt{n}$, where $k_a, k_b$ are known a priori. 
\end{assumption}
This assumption is entirely reasonable; many Euler-Lagrange systems will have an expected operating range for $\vect{q}$ and $\dot{\vect{q}}$.

\begin{theorem}\label{theorem:directed_main_MI}
Suppose that each follower agent with dynamics \eqref{eq:el_dynamics_2}, under Assumption~\ref{assm:no_gravity}, employs the controller \eqref{eq:agent_control_law_direct} with trigger function \eqref{eq:agent_trigger_direct}. Suppose further that the directed graph $\mathcal{G}$ contains a directed spanning tree, with the leader agent $0$ as the root node (and thus with no incoming edges). Then there exists a sufficiently large $\mu$, and sufficiently small $\beta_1, \beta_2$ which ensures that the leader-follower consensus objective is achieved semi-globally exponentially fast.
\end{theorem}
\begin{proof}
The proof of this theorem is lengthy and complex due to the combination of the highly nonlinear Euler-Lagrange dynamics, the directed graph, and the event-based controller. To ensure the presentation of results is not disrupted, we move the proof to Appendix~\ref{app:sec_dir_IC}.
\end{proof}

\section{Adaptive, model-dependent controller on directed network}\label{sec:adaptive}
In this section, we propose an adaptive, distributed event-triggered controller to achieve leader-follower consensus for a directed network of Euler-Lagrange agents. This allows for uncertain parameters in each agent, e.g. the mass of a robotic manipulator arm, and includes the gravitational forces. 

\subsection{Main Result}
Before we present the main results, we introduce variables which allow us to rewrite the multi-agent system in a way which facilitates stability analysis. A lemma on stability is also provided. To begin, we introduce the following auxiliary variables $\vect{q}_{ri}$ and $\vect{s}_i$, which appeared in \cite{mei2013distributed, mei2012distributed} studying leader-follow problems in directed Euler-Lagrange networks. Define
\begin{align}
&\dot{\vect q}_{ri}(t)=-\alpha\sum_{j=0}^{n}a_{ij}(\vect{q}_{i}(t)-\vect{q}_{j}(t)),\\
&\vect{s}_{i}(t)=\dot{\vect q}_{i}(t)-\dot{\vect q}_{ri}(t)=\dot{\vect q}_{i}(t)+\alpha\sum_{j=0}^{n}a_{ij}(\vect{q}_{i}(t)-\vect{q}_{j}(t)), \notag\\
&~~~~~~~~~~~~~~~~~~~~~~~~~~~~~~~~~~~~~~~~~i=1,\ldots,n \label{eq:r3}
\end{align}
where $\alpha$ is a positive constant, $a_{ij}$ is the weighted $(i,j)$ entry of the adjacency matrix $\mathcal{\mat{A}}$ associated with the directed graph $\mathcal{G}$ that characterises the sensing flows among the $n$ followers. According to \textbf{Lemma} \ref{lem:lem5}, one can then verify that the compact form of \eqref{eq:r3} can be written as:
\begin{align}
\dot{\vect q}(t)=-\alpha(\mat{\mathcal L}_{22}\otimes \mat{I}_p)(\vect{q}(t)-\mathbf{1}_{n}\otimes \vect{q}_{0}) + \vect{s}(t)   \label{eq:r4}
\end{align}

The following lemma will later be used for stability analysis of the networked system.
\begin{lem}(From \cite{mei2012distributed}) \label{lem:l2}
Suppose that, for the system~\eqref{eq:r4}, the graph $\mathcal G$ contains a directed spanning tree with the leader as the root vertex. Then system \eqref{eq:r4} is input-to-state stable with respect to input $\vect{s}(t)$. If $\vect{s}(t)\rightarrow\vect{0}_p$ as $t\rightarrow\infty$, then $\dot{\vect q}_{i}(t)\rightarrow \vect{0}_p$ and $\vect{q}_{i}(t)\rightarrow \vect{q}_{0}$ as $t\to\infty$.
\end{lem}
Note that the proof of the above lemma is part of the proof of \textbf{Corollary}~3.7 in \cite{mei2012distributed}.

From {P(6) and the definition of $\dot{q}_{ri}$, we obtain 
\begin{align}
\mat{M}_{i}(\vect{q}_{i})\ddot{\vect q}_{ri}+\mat{C}_{i}(\vect{q}_{i},\dot{\vect q}_{i})\dot{\vect q}_{ri}+\vect{g}_{i}(\vect{q}_{i})&=\mat{Y}_{i}(\vect{q}_{i},\dot{\vect q}_{i},\ddot{\vect q}_{ri},\dot{\vect q}_{ri})\vect{\Theta}_{i},\notag\\
i&=1,\ldots,n\label{eq:r12}
\end{align}
Note that $\vect{\Theta}_{i}$ is an unknown but constant vector for agent $i$. Let $\hat{\vect{\Theta}}_i(t)$ be the estimate of $\vect{\Theta}_{i}$ at time $t$. We update $\vect{\hat{\Theta}}_i(t)$ by the following adaptation law:
\begin{align}
\dot{\hat{\vect{\Theta}}}_{i}(t)=-\mat{\Lambda}_{i}\mat{Y}_{i}^\top(t)\vect{s}_{i}(t),~~~i=1,\ldots,n   \label{eq:r17}
\end{align}
where $\mat{\Lambda}_{i}$ is a symmetric positive-definite matrix.

The control algorithm is now proposed. Let the triggering time sequence of agent $i$ be $t_{0}^{i}, t_{1}^{i}, \ldots, t_{k}^{i},\ldots$ with \mbox{$t_0^i := 0$}. The event-triggered controller for follower agent $i$ is designed as:
\begin{align}
&~~~~~~~~~\vect{\tau}_{i}(t)=-\mat{K}_{i}\vect{s}_{i}(t_{k}^{i})+\mat{Y}_{i}(t_{k}^{i})\hat{\mat{\Theta}}_{i}(t_{k}^{i}),\\
&~~~~~~~~~~~~~~~~~~~~~~~~~~i=1,\ldots,n,~~t\in[t_k^i,t_{k+1}^i)      \label{eq:r9}
\end{align}
where $\mat{K}_{i}>0$ is a symmetric positive definite matrix. It is observed that the control torque remains constant in the time interval $[t_k^i, t_{k+1}^i)$, i.e. $\vect{\tau}_i(t)$ is a piecewise-constant function in time. From the definitions of $\vect{q}_{ri}$ and $\vect{s}_i$, calculations show that the system in \eqref{eq:el_dynamics_1} can be written as
\begin{align}
\mat{M}_{i}(\vect{q}_{i})\dot{\vect s}_{i}(t)+\mat{C}_{i}(\vect{q}_{i},\dot{\vect{q}}_{i})\vect{s}_{i}(t)&=-\mat{K}_{i}\vect{s}_{i}(t_{k}^{i})+\mat{Y}_{i}(t_{k}^{i})\hat{\vect\Theta}_{i}(t_{k}^{i})\notag\\
&~~~-\mat{Y}_{i}(t)\vect{\Theta}_{i}\label{eq:r2}
\end{align}
Before the trigger function is presented, we define two types of measurement errors:
\begin{align}
&\vect{e}_{i}(t)=\vect{s}_{i}(t_{k}^{i})-\vect{s}_{i}(t);\notag\\
&\vect{\varepsilon}_{i}(t)=\vect{Y}_{i}(t_{k}^{i})\hat{\vect{\Theta}}_{i}(t_{k}^{i})-\vect{Y}_{i}(t)\hat{\vect{\Theta}}_{i}(t);   \label{eq:r16}
\end{align}
The trigger function is proposed as follows:
\begin{align}\label{eq:trigger_function_3}
f_{i}(\vect{\varepsilon}_{i}(t),\vect{e}_{i}(t),\mu_{i}(t))&=\|\vect{\varepsilon}_{i}(t)\|+\lambda_{\max}(\mat{K}_{i})\|\vect{e}_{i}(t)\|\notag\\
                                  &~~~~-\frac{\gamma_i}{2}\lambda_{\min}(\mat{K}_{i})\|\vect{s}_{i}(t)\|-\omega_{i}(t)   
\end{align}
where $0<\gamma_i<1$, $\omega_{i}(t)=\sigma_i\sqrt{\lambda_{\min}(\mat{K}_i)}\exp(-\kappa_{i}t)$ with $\sigma_i,\kappa_i>0$. The $k^{th}$ event for agent $i$ is triggered as soon as the trigger condition $f_{i}(\vect{\varepsilon}_{i}(t),\vect{e}_{i}(t))=0$ is fulfilled at $t=t_k^i$. For $t\in[t_k^i, t_{k+1}^i)$, the control input is $\vect{\tau}_i(t) = \vect{\tau}_i(t_k^i)$; the control input is updated when the next event is triggered. Furthermore, every time an event is triggered, and in accordance with their definitions, the measurement errors $\vect{\varepsilon}_i(t)$ and $\vect{e}_i(t)$ are reset to be equal to zero. Thus $f_i(\vect{\epsilon}_i(t), \vect{e}_i(t),\omega_{i}(t)) \leq 0$ for all $t\geq 0$.

We now present our main result.
\begin{theorem}
Consider the multi-agent system \eqref{eq:el_dynamics_1} with control law \eqref{eq:r9}. If $\mathcal{G}$ contains a directed spanning tree with the leader as the root vertex (and thus with no incoming edges), then leader-follower consensus ($\|\vect{q}_i-\vect{q}_0\|\rightarrow 0$ and $\|\dot{\vect{q}}_i\|\rightarrow 0$, $i=1,\ldots,n$) is globally asymptotically achieved as $t\rightarrow\infty$ and no agent will exhibit Zeno behaviour.
\end{theorem}
\begin{proof}
In this part, we focus on the stability analysis of the system \eqref{eq:r2}. The proof on the exclusion of Zeno behaviour is omitted since the idea is the same with that in Part 2) of the proof of \textbf{Theorem} \ref{theorem_2}. Notice that \eqref{eq:r2} is non-autonomous in the sense that it is not self-contained ($\mat{M}_i, \mat{C}_i$ depend on $\vect{q}_i$ and $\vect{q}_i, \dot{\vect{q}}_i$ respectively). However, study of a Lyapunov-like function shows leader-follower consensus is achieved.

We make use of abuse of notation by omitting the argument of time $t$ for time-dependent functions when appropriate, e.g. $\vect{q}_i$ denotes $\vect{q}_i(t)$.

Consider the following Lyapunov-like function
\begin{align}
V=\frac{1}{2}\sum_{i=1}^{N}\vect{s}_{i}^{\top}\mat{M}_{i}(\vect{q}_{i})\vect{s}_{i}+\frac{1}{2}\sum_{i=1}^{N}\vect{\tilde{\Theta}}_{i}^{\top}\mat{\Lambda}_{i}^{-1}\vect{\tilde{\Theta}}_{i}   \label{eq:r7}
\end{align}
where
\begin{align}
\vect{\tilde{\Theta}}_{i}=\vect{\Theta}_{i}-\vect{\hat{\Theta}}_{i} \label{eq:r11}
\end{align}
The derivative of $V$ along the solution of \eqref{eq:r2} is
\begin{align*}
\dot{V}&=\frac{1}{2}\sum_{i=1}^{N}\vect{s}_{i}^{\top}\dot{\mat M}_{i}(\vect{q}_{i})\vect{s}_{i}+\sum_{i=1}^{N}\vect{s}_{i}^{\top}\mat{M}_{i}(\vect{q}_{i})\dot{\vect{s}}_{i}+\sum_{i=1}^{N}\vect{\tilde{\Theta}}^{\top}_{i}\mat{\Lambda}_{i}^{-1}\dot{\vect{\tilde{\Theta}}}_{i}\\
       &=\sum_{i=1}^{N}\vect{s}_{i}^{\top}\left(\frac{1}{2}\dot{\mat M}_{i}(\vect{q}_{i})-\mat{C}_{i}(\vect{q}_{i},\dot{\vect q}_{i})\right)\vect{s}_{i}-\sum_{i=1}^{N}\vect{s}_{i}^{\top}\mat{K}_{i}\vect{s}_{i}(t_{k}^{i})\\
       &~~~+\sum_{i=1}^{N}\vect{s}_{i}^{\top}\mat{Y}_{i}(t_{k}^{i})\vect{\hat{\Theta}}_{i}(t_{k}^{i})-\sum_{i=1}^{N}\vect{s}_{i}^{\top}\mat{Y}_{i}\vect{\Theta}_{i}+\sum_{i=1}^{N}\vect{\tilde{\Theta}}^{\top}_{i}\mat{Y}^{\top}_{i}\vect{s}_{i}\\
\end{align*}
From P(4) we have $\frac{1}{2}\dot{\mat M}_{i}(\vect{q}_{i})-\mat{C}_{i}(\vect{q}_{i})$ is skew-symmetric and with $\vect{\Theta}_i=\vect{\tilde{\Theta}}_i+\vect{\hat{\Theta}}_i$, we obtain
\begin{align*}
\dot{V}&=-\sum_{i=1}^{N}\vect{s}_{i}^{\top}\mat{K}_{i}\vect{s}_{i}(t_{k}^{i})+\sum_{i=1}^{N}\vect{s}_{i}^{\top}\mat{Y}_{i}(t_{k}^{i})\vect{\hat{\Theta}}_{i}(t_{k}^{i})\\
       &~~~-\sum_{i=1}^{N}\vect{s}_{i}^{\top}\mat{Y}_{i}(\vect{\tilde{\Theta}}_{i}+\vect{\hat{\Theta}}_{i})+\sum_{i=1}^{N}\vect{\tilde{\Theta}}^{\top}_{i}\mat{Y}^{\top}_{i}\vect{s}_{i}\\
\end{align*}
By recalling the definition of $\vect{e}_i$ and $\vect{\varepsilon}_i$ in \eqref{eq:r16}, we have
\begin{align*}
\dot{V}
       &=-\sum_{i=1}^{N}\vect{s}_{i}^{\top}\mat{K}_{i}\vect{s}_{i}-\sum_{i=1}^{N}\vect{s}_{i}^{\top}\mat{K}_{i}\vect{e}_{i}+\sum_{i=1}^{N}\vect{s}_{i}^{\top}\vect{\varepsilon}_{i}\\
\end{align*}
Since $\mat{K}_i$ is a symmetric positive definite matrix, the upper bound of $\dot{V}$ is expressed as
\begin{align*}
\dot{V}&\leq -\sum_{i=1}^{N}\lambda_{\min}(\mat{K}_{i})\|\vect{s}_{i}\|^{2}+\sum_{i=1}^{N}\lambda_{\max}(\mat{K}_{i})\|\vect{s}_{i}\|\|\vect{e}_{i}\|\\
       &~~~+\sum_{i=1}^{N}\|\vect{s}_{i}\|\|\vect{\varepsilon}_{i}\|
\end{align*}
Note that the trigger condition $f_{i}(\vect{\varepsilon}_{i}(t),\vect{e}_{i}(t),\omega_i(t))=0$ guarantees that $\|\vect{\varepsilon}_{i}\|+\lambda_{\max}(\mat{K}_{i})\|\vect{e}_{i}\|\leq\frac{\gamma_i}{2}\lambda_{\min}(\mat{K}_{i})\|\vect{s}_{i}\|+\omega_{i}(t)$ holds throughout the evolution of system \eqref{eq:r2}. By further introducing the definition of $\omega_{i}(t)$ in \eqref{eq:trigger_function_3}, we obtain
\begin{align*}
\dot{V} &\leq -\sum_{i=1}^{N}\lambda_{\min}(\mat{K}_{i})\|\vect{s}_{i}\|^{2} + \sum_{i=1}^{N}\frac{\gamma_i}{2}\lambda_{\min}(\mat{K}_{i})\|\vect{s}_{i}\|^{2}\\
        &~~~+ \sum_{i=1}^{N}\sqrt{\lambda_{\min}(\mat{K}_i)}\|\vect{s}_{i}\|\sigma_i\exp(-\kappa_{i}t)\\
\end{align*}
Because there holds $|xy|\leq\frac{\gamma_i}{2}x^2+\frac{1}{2\gamma_i}y^2, \forall x,y\in\mathbb{R}$, for $0 < \gamma_i < 1$, analysis of the right hand side of the above inequality implies that $\dot V$ can be further upper bounded as
\begin{align}
\dot{V} &\leq \sum_{i=1}^{N}(\gamma_i-1)\lambda_{\min}(K_{i})\|\vect{s}_{i}\|^{2} + \sum_{i=1}^{N}\frac{\sigma_i^2}{2\gamma_i}\exp(-2\kappa_i t) \label{eq:r6}
\end{align}
Integrating both sides of \eqref{eq:r6} for any $t>0$ yields:
\begin{align}
&V + \sum_{i=1}^{N}(1-\gamma_i)\lambda_{\min}(\mat{K}_{i})\int_0^t\|\vect{s}_i(\tau)\|^2d\tau\notag\\
&~~~~~~~~~~~~~~~~~~~~~~~~\leq V(0)+ \sum_{i=1}^{N}\frac{\sigma_i^2}{4\gamma_i\kappa_i}\label{eq:r8}
\end{align}
which implies that $V$ is bounded. Since $V$ is bounded, according to \eqref{eq:r7}, both $s_i$ and $\vect{\tilde{\Theta}}_i(t)$, for all $i\in \{1, ..., n\}$, are bounded. Now we return to \eqref{eq:r12} and obtain that
\begin{align*}
\|\mat{Y}_{i}\vect{\Theta}_{i}\|\leq\|\mat{M}_{i}\|\|\ddot{\vect{q}}_{ri}\|+\|\mat{C}_{i}\|\|\dot{\vect{q}}_{ri}\|+\|\vect{g}_{i}\|
\end{align*}
By recalling that the linear system \eqref{eq:r4} is input-to-state stable and the fact that $s$ is bounded, we conclude that $\vect{q}_i$ and $\dot{\vect{q}}_i$ are both bounded. Because $\vect{q}_i$ and $\dot{\vect{q}}_i$ are bounded then, from their definitions, so are $\dot{\vect q}_{ri}$ and $\ddot{\vect q}_{ri}$. Then from P(2), P(3) and P(5), the assumed properties of Euler-Lagrange equations, we have that $\|\mat{Y}_i\|$ is upper bounded by a positive value. From the above conclusions, it is straightforward to see that the right hand side of \eqref{eq:r2}, $\mat{M}_i$, $\mat{C}_i$ and $\vect{s}_i$ are all bounded. We thus obtain that $\dot{\vect s}_i$ is bounded. From this, it is obvious that $\vect{s}_i, \dot{\vect{s}}_i\in L_\infty$. Turning to \eqref{eq:r8}, it follows that
\begin{align}
\sum_{i=1}^{N}(1-\gamma_i)\lambda_{\min}(\mat{K}_{i})\int_0^t\|\vect{s}_i(\tau)\|^2d\tau \leq V(0) + \sum_{i=1}^{N}\frac{\sigma_i^2}{4\gamma_i\kappa_i}
\end{align}
which indicates that $\int_0^t\|s_i(\tau)\|^2d\tau$ is bounded and thus $s_i\in \mathcal{L}_2$. By applying \textbf{Lemma} \ref{lem:l1}, we have that $s_i\rightarrow \vect{0}_p$ as $t\rightarrow\infty$. Then by applying \textbf{Lemma} \ref{lem:l2}, we conclude that $\vect{q}_i-\vect{q}_0\rightarrow\vect{0}_p$ and $\dot{\vect{q}}_i\rightarrow\vect{0}_p$ as $t\rightarrow\infty$. The leader-follower objective is globally asymptotically achieved.

\end{proof}

\section{Simulations}\label{sec:simulation}
In this subsection, we will provide three simulations to respectively demonstrate the performance of the three proposed controllers in this paper for application to industrial manipulators (See Fig. \ref{fig:manipulator}). Although the effectiveness of controller \eqref{eq:control_input_1} is verified in subsection \ref{sec3-b}, the applied system is simple one-arm manipulators. In this subsection we assume that all two-link manipulators share the same dynamic models and parameters. The Euler-Lagrange equations for the $i^{\text{th}}$ two-link manipulator is: 
\begin{equation*}
\begin{bmatrix} 
M^{11}_i & M^{12}_i\\
M^{21}_i & M^{22}_i\\
\end{bmatrix}
\begin{bmatrix}
\ddot{q}^{(1)}_{i}\\
\ddot{q}^{(2)}_{i}\\
\end{bmatrix}
+\begin{bmatrix}
C^{11}_i & C^{12}_i\\
C^{21}_i & C^{22}_i\\
\end{bmatrix}
\begin{bmatrix}
\dot{q}^{(1)}_{i}\\
\dot{q}^{(2)}_{i}\\
\end{bmatrix}
+\begin{bmatrix}
g^{(1)}_{i}\\
g^{(2)}_{i}\\
\end{bmatrix}
=\begin{bmatrix}
\tau^{(1)}_{i}\\
\tau^{(2)}_{i}\\
\end{bmatrix}
\end{equation*}
The elements in $M_i$, $C_i$ matrices and $g_i$ vector are given below:
\begin{equation*}
\begin{split}
M^{11}_i &= (m_1+m_2)d_{1}^{2} + m_2d_{2}^{2} + 2m_2d_1d_2\cos(q_{i}^{(2)})\\
M^{12}_i &= M^{21}_i = m_2(d_{2}^{2}+d_1d_2\cos(q_{i}^{(2)}))\\
M^{22}_i &= m_2d_{2}^{2}\\
C^{11}_i &= -m_2d_1d_2\sin(q_{i}^{(2)})\dot{q}_{i}^{(2)}\\
C^{12}_i &= -m_2d_1d_2\sin(q_{i}^{(2)})\dot{q}_{i}^{(2)} - m_2d_1d_2\sin(q_{i}^{(2)})\dot{q}_{i}^{(1)}\\
C^{21}_i &= m_2d_1d_2\sin(q_{i}^{(2)})\dot{q}_{i}^{(1)}\\
C^{22}_i &= 0\\
g^{(1)}_{i} &= (m_1+m_2)gd_1\sin(q_{i}^{(1)}) + m_2gd_2\sin(q_{i}^{(1)}+q_{i}^{(2)})\\
g^{(2)}_{i} &= m_2gd_2\sin(q_{i}^{(1)}+q_{i}^{(2)})\\
\end{split}
\end{equation*}
where $g$ is the acceleration due to gravity, $d_1$ and $d_2$ are lengths of the $1^{st}$ and $2^{nd}$ links of the manipulator, respectively; $m_1$ and $m_2$ are mass of the $1^{st}$ and $2^{nd}$ of the manipulator. The physical parameters of each manipulator are selected as $g=9.8\,\mathrm{m/s^2}$, $d_1=1.5\,\mathrm{m}$, $d_2=1\,\mathrm{m}$, $m_1=1\,\mathrm{kg}$, $m_2=2\,\mathrm{kg}$. The initial states of each manipulator are shown in Table I. Note that in simulations 1 and 2, we assume that $g^{(1)}_{i}, g^{(2)}_{i}=0$.

\newcounter{mytempeqncnt}
\begin{figure*}[!t]
\normalsize
\begin{align*}
&\Theta_i=
\begin{bmatrix}
(m_1+m_2)d_{1}^{2} + m_2d_{2}^{2} & m_2d_1d_2 & m_2d_{2}^{2} & (m_1+m_2)gd_1 &m_2gd_2
\end{bmatrix}^T\\
&Y_i=
\begin{bmatrix}
x_1 & 0\\
2x_1\cos(q_{i}^{(2)})+x_2\cos(q_{i}^{(2)})-y_1\sin(q_{i}^{(2)})\dot{q}_{i}^{(2)}-y_2\sin(q_{i}^{(2)})\dot{q}_{i}^{(2)}-y_2\sin(q_{i}^{(2)})\dot{q}_{i}^{(1)} & x_1\cos(q_{i}^{(2)})+y_1\sin(q_{i}^{(2)})\dot{q}_{i}^{(1)}\\
x_2 & x_1+x_2\\
\sin(q_{i}^{(1)}) & 0\\
\sin(q_{i}^{(1)}+q_{i}^{(2)}) & \sin(q_{i}^{(1)}+q_{i}^{(2)})\\
\end{bmatrix}^T
\end{align*}
\hrulefill
\vspace*{4pt}
\end{figure*}

\begin{figure}[htb]
\centering{
\resizebox{55mm}{!}{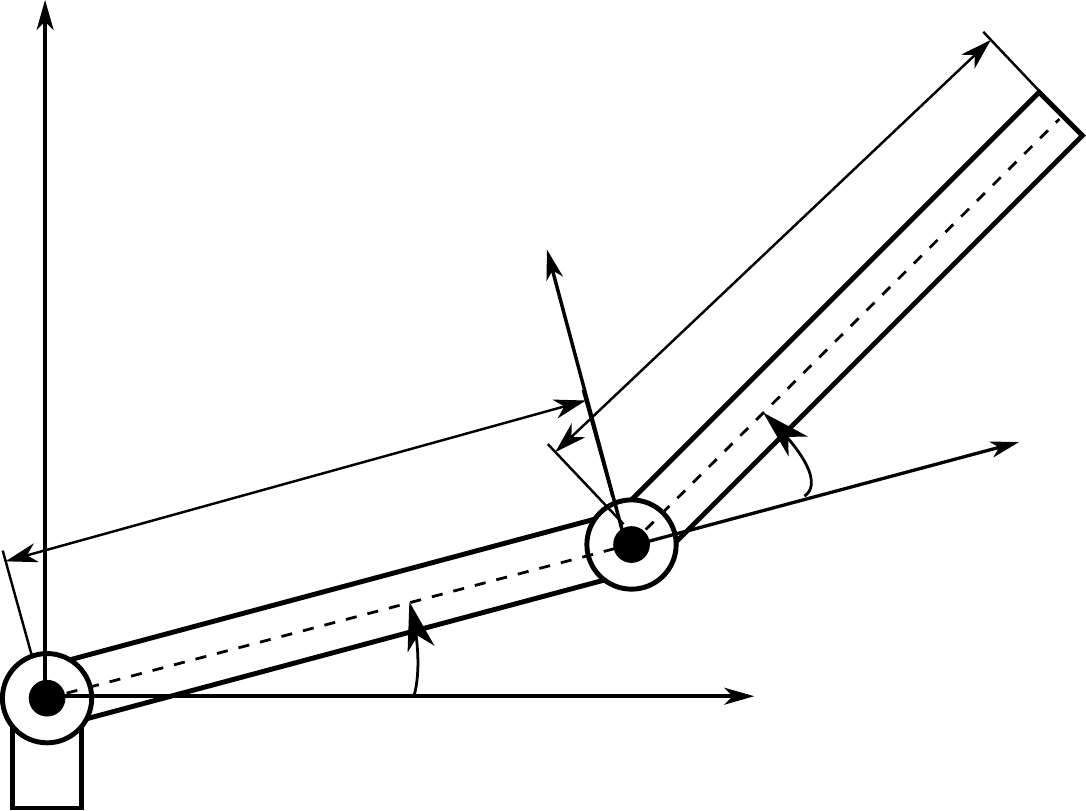}
\caption{Two-link manipulator, generalized coordinates $\vect{q}=[q^1,q^2]^\top$}.  
\label{fig:manipulator}
}
\end{figure}

\begin{table}
\caption{Agents' initial states used in simulations}
\label{tab:agent_param}
\begin{center}
{\renewcommand{\arraystretch}{1.1}
\begin{tabular}{c c c c c}
• & $q_i^{(1)}(0)$ & $q_i^{(1)}(0)$ & $\dot{q}_i^{(1)}(0)$ & $\dot{q}_i^{(1)}(0)$  \\\hline
Agent 0 & $\pi/6$ & $\pi/3$ & 0.0 & 0.0  \\\hline
Agent 1 & $\pi/5$ & $\pi/6$ & 0.8 & 0.2  \\\hline
Agent 2 & $\pi/6$ & $\pi/4$ & -0.2 & 0.3  \\\hline
Agent 3 & $\pi/9$ & $\pi/6$ & 0.6 & -0.4 \\\hline
Agent 4 & $\pi/8$ & $\pi/4$ & 0.5 & 0.1 \\\hline
Agent 5 & $\pi/9$ & $\pi/6$ & 0.1 & 0.0 \\\hline
\end{tabular}
}
\end{center}
\end{table}

\begin{simulation} \label{sim:1}
This simulation will demonstrate the performance of controller \eqref{eq:control_input_1} under trigger condition \eqref{eq:trigger_function_1}. The sensing graph $\mathcal{G}$ associated with the five follower manipulators and the leader manipulator has the following weighted Laplacian
\begin{equation*}
\mathcal{L}=
\begin{bmatrix}
0 & 0 & 0 & 0 & 0 & 0\\
-1 & 3.9 & -1.55 & 0 & -1.35 & 0\\
-1 & -1.55 & 6.4 & -2.1 & -1.75 & 0\\
0 & 0 & -2.1 & 7.35 & -2.35 & -2.9\\
0 & -1.35 & -1.75 & -2.35 & 6.7 & -1.25\\
0 & 0 & 0 & -2.9 & -1.25 & 4.15\\
\end{bmatrix}
\end{equation*}
The initial value of the variable-gain scalar $\mu_(t)$ is chosen as $\mu_i(0)=0$. The exponential function used in trigger function \eqref{eq:trigger_function_1} is selected as $\omega_i(t)=1.8*\exp(-0.2*t)$. The performance of the controller is demonstrated in Fig. \ref{fig:controller_1}.
\end{simulation}

\begin{simulation} \label{sim:2}
The directed sensing graph $\mathcal{G}$ associated with the five follower manipulators and the leader manipulator has the following weighted Laplacian
\begin{equation}\label{eq:directed_laplacian}
\mathcal{L}=
\begin{bmatrix}
0 & 0 & 0 & 0 & 0 & 0\\
-1 & 2.55 & -1.55 & 0 & 0 & 0\\
-1 & 0 & 3.1 & -2.1 & 0 & 0\\
0 & 0 & -2.1 & 2.1 & 0 & 0\\
0 & -1.35 & -1.75 & -2.35 & 5.45 & 0\\
0 & 0 & 0 & -2.9 & -1.25 & 4.15\\
\end{bmatrix}
\end{equation}
and it contains a directed spanning tree rooted at $v_0$. There are no incoming edges to $v_0$. The gain $\mu$ used in controller \eqref{eq:agent_control_law_direct} is selected as $4$. The parameters $\beta_1$ and $\beta_2$ introduced in trigger function \eqref{eq:agent_trigger_direct} are both chosen as $0.6$. The performance of the proposed control algorithm is shown in Fig. \ref{fig:controller_2}. 
\end{simulation}

\begin{simulation} \label{sim:3}
The uncertain parameter vector $\Theta_i$ for each manipulator and the regression matrix are given at the top of the next page. Note that $\Theta_i$ is unknown for manipulator $i$. The Laplacian matrix is also chosen as \eqref{eq:directed_laplacian}. The control gain matrix (in \eqref{eq:r9}) for all follower manipulators is chosen as $K_i=\emph{I}_{2}$. The parameter $\gamma_i$ required in the trigger function \eqref{eq:trigger_function_3} is selected as $0.6$ for manipulator $i=1,\ldots,5$. Lastly, $\mu_i(t)=5\exp(-0.6t)$ is used in the trigger functions for all manipulators in the follower network.
The performance of controller \eqref{eq:r9} with trigger function \eqref{eq:trigger_function_3} is presented in Fig. \ref{fig:controller_3}.
\end{simulation}

\begin{figure*}[!ht]
\begin{center}
\includegraphics[width=1.05\linewidth]{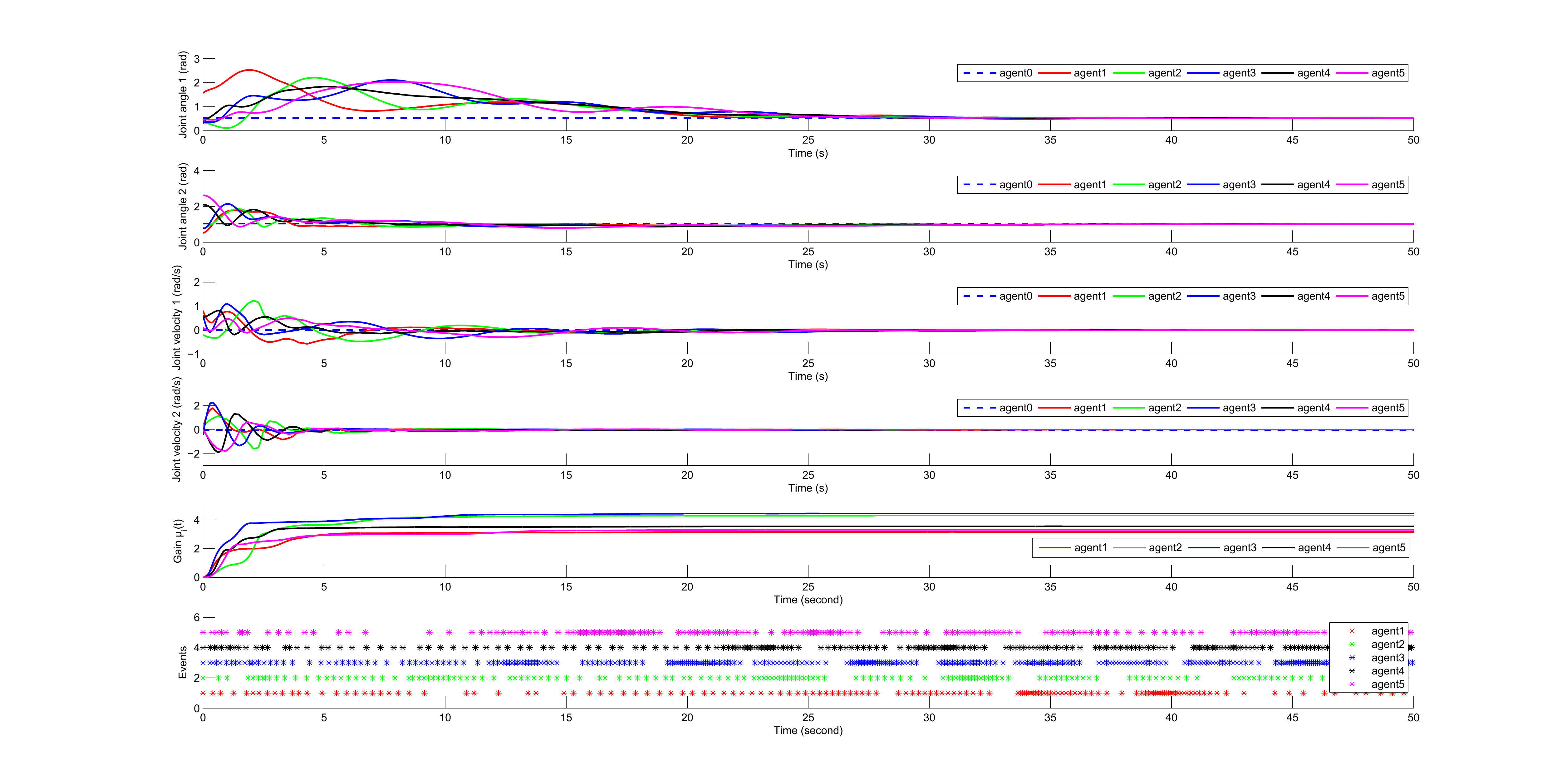}
\caption{Simulation results for controller \eqref{eq:control_input_1} under trigger function \eqref{eq:trigger_function_1}. From top to bottom: the plots the generalized coordinates; the plots of generalised velocities of all the follower manipulators; the plot of variable gain $\mu_i(t)$; the plot of trigger events}
\label{fig:controller_1}
\end{center}
\end{figure*}

\begin{figure*}[!ht]
\begin{center}
\includegraphics[width=1.05\linewidth]{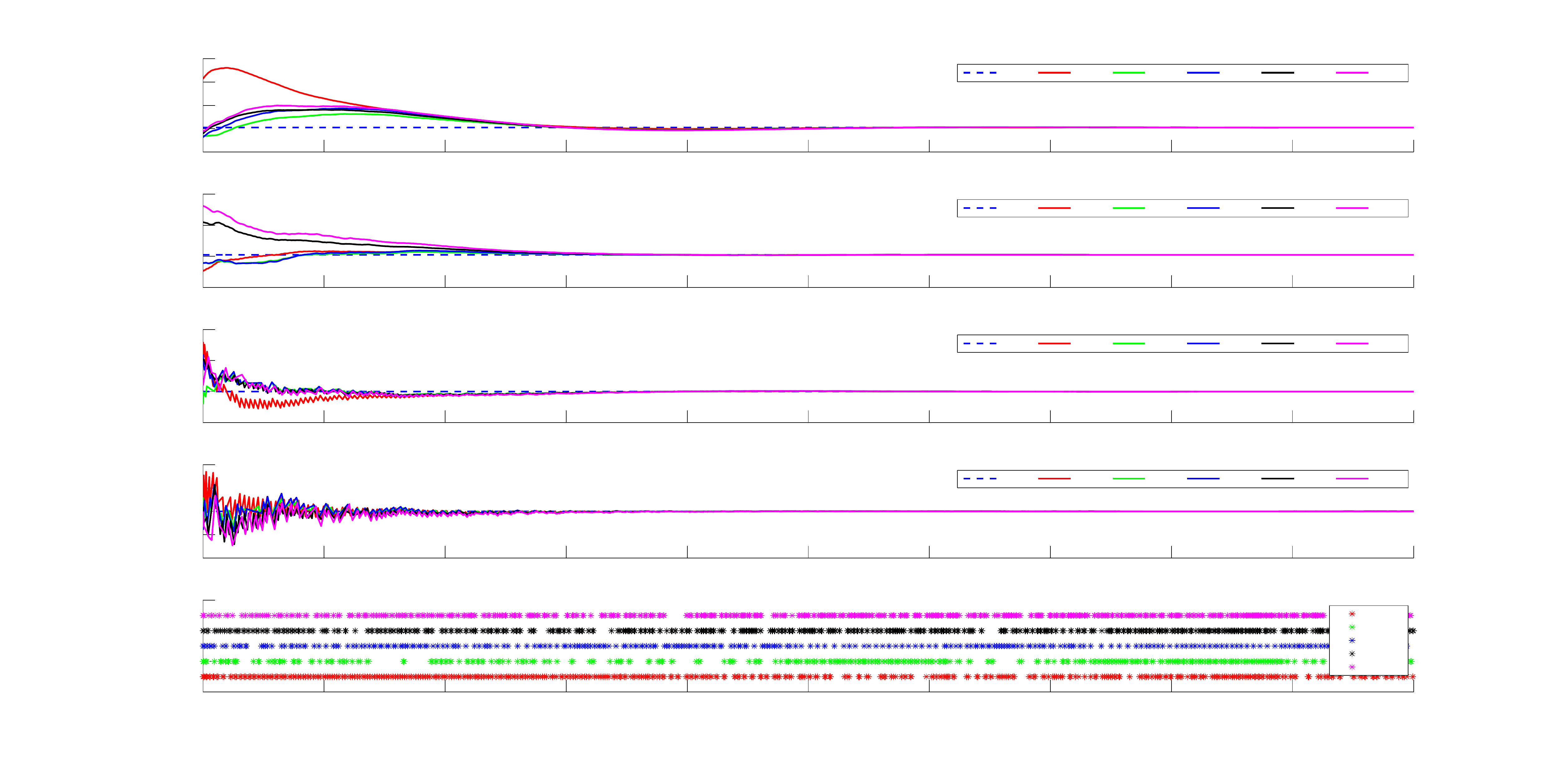}
\caption{Simulation results for controller \eqref{eq:agent_control_law_direct} under trigger function \eqref{eq:agent_trigger_direct}. From top to bottom: the plots the generalized coordinates; the plots of generalised velocities of all the follower manipulators; the plot of trigger events}
\label{fig:controller_2}
\end{center}
\end{figure*}

\begin{figure*}[!ht]
\begin{center}
\includegraphics[width=1.05\linewidth]{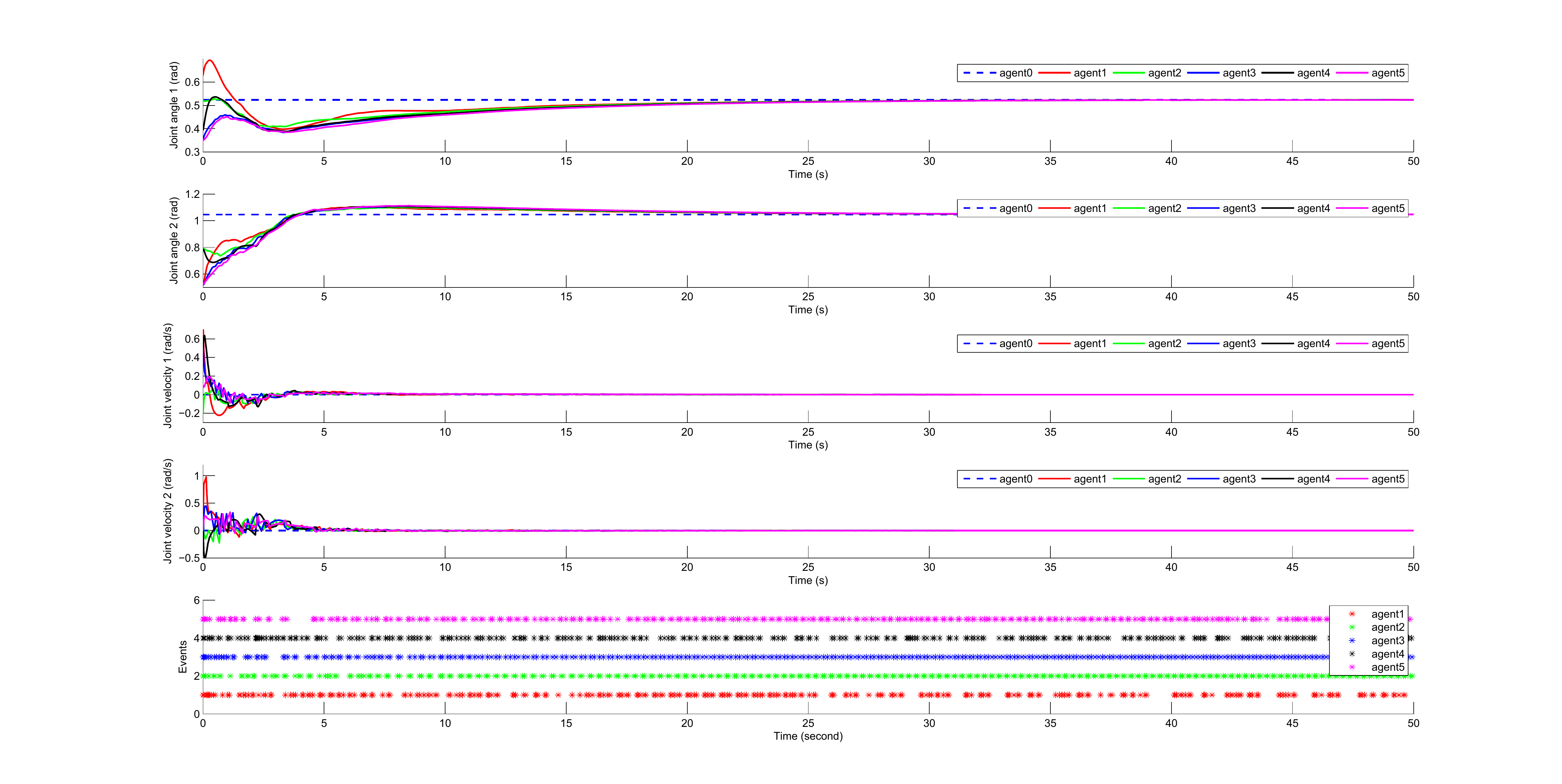}
\caption{Simulation results for controller \eqref{eq:r9} under trigger function \eqref{eq:trigger_function_3}. From top to bottom: the plots the generalized coordinates; the plots of generalised velocities of all the follower manipulators; the plot of trigger events}
\label{fig:controller_3}
\end{center}
\end{figure*}

\section{Conclusions}\label{sec:conclusion}
This paper proposed and showed the stability of three different algorithms for achieving leader-follower consensus for a network of Euler-Lagrange agents. Each algorithm is suited for a different scenario and have their advantages and disadvantages, and can be chosen depending on the problem requirements. For each algorithm, we propose a mixed trigger function containing an error term (a fundamental part of event-based control), a state-dependent term (to prevent frequent/dense triggering of events as time goes to infinity), and an exponentially decaying offset (to guarantee exclusion of Zeno behaviour in finite time intervals). The effectiveness of such a mixed trigger function is extensively explained via simulations. Additional simulations verify the effectiveness of each algorithm. Future work includes relaxation of the continuous sensing requirement to allow for event-based sensing. One possible approach is via self-triggered controllers, but may be difficult due to the complex Euler-Lagrange dynamics. A second key future work is to obtain a constant lower bound on the inter-event times (currently the lower bound decreases as time increases). This might be achieved through additional analysis or modification of the offset function.

\section*{Acknowledgements}
The authors would like to acknowledge Brian D.O. Anderson for his invaluable discussions on stability of non-autonomous systems.

\bibliographystyle{ieeetr}
\bibliography{bib_event_lagrange,MYE_ANU}

\appendices

\section{Proofs for Section~\ref{sec:background}}\label{app:sec_background}
\subsection{\textbf{Lemma}~\ref{lem:bounded_pd_function} } 
Observe that $cxy^2 \leq c\mathcal{X}y^2$ for all $x \in [0,\mathcal{X}]$. It follows that 
\begin{equation*}
g(x,y) \geq a x^2 + (b - c\mathcal{X}) y^2 - d x y
\end{equation*} 
if $y \in [0, \infty)$ and $x \in [0,\mathcal{X}]$ because $c > 0$. For any fixed value of $y = y_1 \in [0,\infty)$, write $\bar g(x) = a x^2 + (b - c\mathcal{X}) y_1^2 - d x y_1$. The discriminant of $\bar g(x)$ is negative if 
\begin{equation}\label{eq:b_bound_lemma_pd}
b > c\mathcal{X} + \frac{d^2}{4a}
\end{equation}
which implies that the roots of $\bar g(x)$ are complex, i.e. $\bar g(x) > 0$ and this holds for any $y_1 \in [0,\infty)$. We thus conclude that for all $y \in [0, \infty)$ and $x \in [0,\mathcal{X}]$, if $b$ satisfies \eqref{eq:b_bound_lemma_pd}, then $g(x,y) > 0$ except the case where $g(x,y) = 0$ if and only if $x = y = 0$. 

\subsection{\textbf{Corollary} \ref{cor:bounded_pd_function_2}}
Observe that $h(x,y) = g(x,y) - e x -f y$ where $g(x,y)$ is defined in \textbf{Lemma}~\ref{lem:bounded_pd_function}. Let $b^*$ be such that it satisfies condition \eqref{eq:b_bound_lemma_pd} in \textbf{Lemma} \ref{lem:bounded_pd_function} and thus $g(x,y) > 0$ for $x\in [0,\infty)$ and $y\in[0,\mathcal{Y}]$. Note that the positivity condition on $g(x,y)$ in \textbf{Lemma} \ref{lem:bounded_pd_function} continues to hold for any $a\geq a^*$ and any $b \geq b^*$. Let $a_1$ and $b_1$ be positive scalars whose magnitudes will be determined later. Define $a = a_1 + a^*$ and $b = b_1 + b^*$. Define $z(x,y) \triangleq a_1 x^2 + b_1 y^2 - e x - f y$.  Next, consider $(x,\bar y) \in \mathcal{V}$, where $\bar y$ is some fixed value. It follows that
\begin{equation*}
z(x, \bar y) = a_1 x^2 -  e x + (b_1\bar y^2-f\bar y)
\end{equation*} 
Note the discriminant of $z(x,\bar y)$ is $\mathcal{D}_x =  e^2 - 4 a_1 ( b_1 \bar y^2 - f \bar y)$. It follows that $\mathcal{D}_x < 0$ if $b_1 \bar y^2 > f\bar y + e/4a_1$. This is satisfied, independently of $\bar y \in [\mathcal{Y}-\varepsilon, \mathcal{Y}]$, for any $b_1 \geq b_{1,y}$, $a_1 \geq a_{1,y}$ where 
\begin{equation*}
b_{1,y} > \frac{ e^2 }{4a_{1,y} (\mathcal{Y}- \varepsilon)^2 } + \frac{f}{\mathcal{Y}-\varepsilon}
\end{equation*}
because $\mathcal{Y}-\varepsilon \leq \bar y$. It follows that $\mathcal{D}_x < 0 \Rightarrow z(x,y) > 0$ in $\mathcal{V}$. Now, consider $(\bar x,y) \in \mathcal{U}$ for some fixed value $\bar x$. It follows that
\begin{equation*}
z(\bar x,y) = b_1 y^2 - f y + (a_1 \bar x^2 - e\bar x)
\end{equation*}
and note the discriminant of $z(\bar x,y)$ is $\mathcal{D}_y = f^2- 4 b_1 (a_1\bar x^2-e\bar x)$. Suppose that $a_1 > e/\mathcal{X}$, which ensures that $a_1\bar x^2-e\bar x > 0$. Then, $\mathcal{D}_y < 0$ if $b_1(a_1\bar x^2 - e\bar x) > f/4$. This is satisfied, independently of $\bar x \in [\mathcal{X} - \vartheta, \mathcal{X}]$, for any $b_1 \geq b_{1,x}$, $a_1 \geq a_{1,x}$ where
\begin{equation*}
b_{1,x} > \frac{f}{4(a_{1,x} (\mathcal{X} - \vartheta)^2 - e (\mathcal{X} - \vartheta))}
\end{equation*}
It follows that $\mathcal{D}_y < 0 \Rightarrow z(x,y) > 0$ in $\mathcal{U}$. We conclude that setting \mbox{$b = b^* + \max [b_{1,x}, b_{1,y}]$} and \mbox{$a = a^* + \max [a_{1,x}, a_{1,y}]$}, implies $h(x, y) > 0$ in $\mathcal{R}$, except $h(0,0) = 0$.


\section{Proofs for Section~\ref{sec:dir_IC}}\label{app:sec_dir_IC}
Before we present the main proof of \textbf{Theorem}~\ref{theorem:directed_main_MI}, we need to compute an upper bound using limited information about the initial conditions.

\subsection{An Upper Bound Using Initial Conditions}\label{subsec:IC_bound} 

Suppose that initial conditions are bounded as $\Vert \vect u(0) \Vert \leq k_a$ and $\Vert \vect v(0) \Vert \leq k_b$ with $k_a, k_b$ known a priori. Before we proceed with the main proof, we provide a method to calculate a non-tight upper bound on the initial states expressed as $\Vert \vect u(0) \Vert < \mathcal{X}$ and $\Vert \vect v(0) \Vert < \mathcal{Y}$, with the property that as shown in the sequel, there holds $\Vert \vect u(t) \Vert < \mathcal{X}$ and $\Vert \vect v(t) \Vert < \mathcal{Y}$ for all $t \geq 0$, and exponential convergence results. Due to spatial limitations, we show only the bound on $\vect v$ and leave the reader to follow an identical process for $\vect u$. In keeping with the model-independent nature of the paper, define a function as
\begin{equation}\label{eq:V_bar}
\bar V_{\mu} = \begin{bmatrix}\vect u \\ \vect v\end{bmatrix}^\top \!
\begin{bmatrix}
\lambda_{\max}(\mat Q)\mat I_{np} & \frac{1}{2}\mu^{-1}\bar\gamma(k_{\overline{M}}\!+\!\delta)\mat I_{np} \\ \frac{1}{2}\mu^{-1}\bar\gamma(k_{\overline{M}}\!+\!\delta)\mat I_{np} & \frac{1}{2}\bar\gamma(k_{\overline{M}}\!+\!\delta)\mat I_{np}
\end{bmatrix} 
\begin{bmatrix}\vect u \\ \vect v\end{bmatrix}
\end{equation} 
where $\mat Q = \mat \Gamma \mathcal{L}_{22} + \mathcal{L}_{22}^\top \mat \Gamma$, $\underline{\gamma} = \min_i \gamma_i$ and $\bar\gamma = \max_i \gamma_i$. Here, $\gamma_i$ are the diagonal entries of $\mat \Gamma_p$. The constant $\delta > 0$ is arbitrarily small and fixed. Note that \mbox{$(k_{\overline{M}}+\delta)\mat{I}_{np} > \mat M$} and that $\bar V_{\mu}$ is \emph{not} a Lyapunov function. Let the matrix in \eqref{eq:V_bar} be $\mat L_\mu$. Then according to \textbf{Theorem}~\ref{theorem:schur_complement}, $\mat L_\mu > 0$ if and only if $\lambda_{\max}(\mat Q)\mat I_{np} - \frac{1}{2}\mu^{-2}\bar\gamma(k_{\overline{M}}+ \delta)\mat I_{np} > 0 $ which is implied by $\lambda_{\max}(\mat Q) - \frac{1}{2}\mu^{-2}\bar\gamma(k_{\overline{M}}+ \delta) > 0$. Then $\mat L_\mu > 0$ for any $\mu \geq \mu_1^*$ where
\begin{equation*}
\mu_1^* > \sqrt{ \frac{\bar\gamma(k_{\overline{M}}+ \delta)}{2\lambda_{\max}(\mat Q) } }
\end{equation*}
While $\bar V_\mu$ is a function of $\vect u(t)$ and $\vect v(t)$, we use $\bar V_\mu(t)$ to denote $\bar V_\mu (\vect u(t), \vect v(t))$. Lastly, observe that 
\begin{align*}
\bar V_\mu & \leq \lambda_{\max}(\mat Q) \Vert \vect u \Vert^2 + \frac{1}{2}\bar\gamma(k_{\overline{M}}+\delta) \Vert \vect v \Vert^2 \\
& \; + \mu^{-1}\bar\gamma(k_{\overline{M}}+\delta) \Vert \vect u \Vert \Vert \vect v \Vert
\end{align*}

Next, let
\begin{equation}\label{eq:V_LB}
\underline{V}_{\mu} = \begin{bmatrix}\vect u \\ \vect v\end{bmatrix}^\top \!
\begin{bmatrix}
\frac{1}{4} \lambda_{\min}(\mat Q)\mat I_{np} & \frac{1}{2}\mu^{-1}\underline{\gamma}(k_{\underline{m}}\!-\!\delta)\mat I_{np} \\ \frac{1}{2}\mu^{-1}\underline{\gamma}(k_{\underline{m}}\!-\!\delta)\mat I_{np} & \frac{1}{2} \underline{\gamma} (k_{\underline{m}}\!-\!\delta)\mat I_{np}
\end{bmatrix} 
\begin{bmatrix}\vect u \\ \vect v\end{bmatrix}
\end{equation}
Call the matrix in \eqref{eq:V_LB} $\mat N_\mu$. Let the arbitrarily small $\delta$ be such that $(k_{\underline{m}}-\delta) > 0$. Analysis using \textbf{Theorem}~\ref{theorem:schur_complement}, similar to above, is used to conclude that $\mat N_\mu > 0$ for any $\mu \geq \mu_2^*$ where
\begin{equation*}
\mu_2^* > \sqrt{ \frac{2\underline{\gamma}( k_{\underline{m}}-\delta)}{\lambda_{\min}(\mat Q) } }
\end{equation*}  
Set $\mu^*_3 = \max\{\mu^*_1, \mu^*_2\}$. Define 
\begin{align*}
\rho_1(\mu) & = \frac{1}{2} (k_{\overline{M}}+\delta) - \frac{1}{4}\mu^{-2}({k_{\overline{M}}} + \delta)^2\lambda_{\max}(\mat Q)^{-1} \\
\rho_2(\mu) & = \frac{1}{2}(k_{\underline{m}}-\delta) - \mu^{-2}(k_{\underline{m}}-\delta)^2\lambda_{\min}(\mat Q)^{-1} \\
\rho_3(\mu) & = \frac{1}{2}(k_{\underline{m}}-\delta) - \frac{1}{2}\mu^{-2}(k_{\overline{M}})^2\Vert \mat Q^{-1}\Vert
\end{align*} 
and observe that for sufficiently large $\mu$ there holds $\rho_1 \geq \rho_3 > \rho_2$. Assume without loss of generality that $\rho_1 \geq \rho_3 > \rho_2$ (if not, one can always replace $\mu_3^*$ by a $\mu_4^*$ with $\mu^*_4 > \mu^*_3$, and such that $\rho_1 \geq \rho_3 >\rho_2$). Note that for any $\mu\geq \mu^*_3$ there holds $\rho_i(\mu^*_3) \leq \rho_i(\mu), i = 1,2$. Compute now
\begin{align*}
\bar V^* = \lambda_{\max}(\mat Q) k_a^2 + \frac{1}{2}\bar\gamma(k_{\overline{M}}+\delta) k_b^2 + {\mu_3^*}^{-1}\bar\gamma(k_{\overline{M}}+\delta) k_a k_b
\end{align*}
One can verify that for any $\mu \geq \mu_3^*$ that $\bar V_\mu(0) \leq \bar V^*$. It follows from \textbf{Lemma}~\ref{lem:W_bound} and \eqref{eq:W_bound_y} that 
\begin{equation}\label{eq:semi_global_barV_delta}
\Vert \vect v(0)\Vert_2 \leq \sqrt{\frac{\bar V_{\mu}(0)}{\rho_1(\mu)}} \leq \sqrt{\frac{\bar V_{\mu}(0)}{\rho_1(\mu^*_3)}} < \sqrt{\frac{\bar V^*}{\rho_2(\mu^*_3)}} := \mathcal{Y}_1
\end{equation}  
Follow a similar method to obtain $\mathcal{X}_1$. Next, compute
\begin{align*}
\wh{V}^* = \lambda_{\max}(\mat Q) {\mathcal{X}_1}^2 + \frac{1}{2}\bar\gamma(k_{\overline{M}}+\delta) {\mathcal{Y}_1}^2 + {\mu_3^*}^{-1}\bar\gamma(k_{\overline{M}}+\delta) \mathcal{X}_1 \mathcal{Y}_1
\end{align*}
Finally, compute the bound $\mathcal{Y} =\sqrt{\wh{V}^*/\rho_2(\mu_3^*)}$, and note that $\bar V_{\mu_3^*} \leq \wh{V}^*$. Note also that $\bar V_{\mu}, \bar{V}^*, \rho_2(\mu^*_3)$ are independent of $\mu$. Thus $\Vert \vect v(0)\Vert_2 < \mathcal{Y}$  (and similarly $\Vert \vect u(0)\Vert_2 < \mathcal{X}$) can be used for all $\mu\geq \mu^*_3$. We now proceed to the proof of \textbf{Theorem}~\ref{theorem:directed_main_MI}.

\subsection{Proof of \textbf{Theorem}~\ref{theorem:directed_main_MI}}
\emph{Part 1:} Consider the Lyapunov-like candidate function
\begin{equation}\label{eq:lyapunov_candidate}
V = \frac{1}{2}\vect{u}^\top \mat{Q}\vect{u}+\frac{1}{2}\vect{v}^\top\mat \Gamma_p \mat{M}\vect{v}+\mu^{-1}\vect{u}^\top\mat \Gamma_p \mat{M}\vect{v}
\end{equation}
where $\mat \Gamma_p = \mat \Gamma \otimes \mat I_p$. It may also be expressed as a quadratic in the variables $\vect{u}$ and $\vect{v}$
\begin{equation}\label{eq:lyapunov_candidate_quad}
V = \begin{bmatrix}\vect{u} \\ \vect{v}\end{bmatrix}^\top
\begin{bmatrix}\frac{1}{2} \mat{Q} & \frac{1}{2}\mu^{-1}\mat \Gamma_p \mat{M} \\ \frac{1}{2}\mu^{-1}\mat \Gamma_p\mat{M} & \frac{1}{2}\mat \Gamma_p\mat{M}\end{bmatrix}
\begin{bmatrix}\vect{u} \\ \vect{v}\end{bmatrix}
\end{equation}
\textbf{Theorem}~\ref{theorem:schur_complement} is used to conclude that $V$ is positive definite if
\begin{equation}\label{eq:cond_mu_Va}
\frac{1}{2}\mat{Q}-\frac{1}{2}\mu^{-2}\mat \Gamma_p \mat{M} > 0
\end{equation}
which is implied by
\begin{equation}\label{eq:cond_mu_Vaa}
\lambda_{\min}(\mat{Q}) - \mu^{-2}\bar\gamma\lambda_{\max}(\mat{M}) > 0
\end{equation}
Observe that \eqref{eq:cond_mu_Vaa} is implied by
\begin{equation}\label{eq:cond_mu_Vb}
\mu > \sqrt{ \frac{\bar\gamma k_{\overline{M}}}{\lambda_{\min}(\mat{Q})} }
\end{equation}
because there holds $\lambda_{\max}(\mat{M}) \leq k_{\overline{M}}$. Since $\lambda_{\min}(\mat{Q}) > 0$ then there can always be found a $\mu > 0$ which satisfies $\eqref{eq:cond_mu_Vb}$. Define $\mu_5^*$ such that $\mu_5^*$ satisfies \eqref{eq:cond_mu_Vb} and $\mu_5^* \geq \mu_3^*$. Therefore $V$ is positive definite in $\vect{u}$ and $\vect{v}$ for all $\mu \geq \mu_5^*$. Denote the matrix in \eqref{eq:lyapunov_candidate_quad} as $\mat G$. Following the method outlined in the appendix of \cite{ye2016EL_consensus}, it is straightforward to show that $V(t) < \bar{V}_\mu(t)$ for all $t$ because $\mat L_\mu > \mat G > \mat N_\mu$ for all $\mu \geq \mu_5^*$. Lastly, observe that
\begin{align}\label{eq:lyapunov_upperbound}
V(t) & \leq \frac{1}{2}\lambda_{\max}(\mat Q) \Vert \vect u(t) \Vert^2 + \frac{1}{2}\bar\gamma k_{\overline{M}} \Vert \vect v(t) \Vert^2 \nonumber \\
&\; + \mu^{-1}\bar\gamma k_{\overline{M}} \Vert \vect u(t) \Vert \Vert \vect v(t) \Vert
\end{align}
Taking the derivative of $V$ with respect to time along the trajectories of the system \eqref{eq:euler_Lagrange_nonautonomous_network}, we have
\begin{align}
\dot{V} & = \vect{u}^\top\mat Q\vect{v} +\vect{v}^\top\mat \Gamma_p \mat{M}\dot{\vect{v}} +\frac{1}{2}\vect{v}^\top\mat \Gamma_p \dot{\mat{M}}\vect{v} \nonumber \\
& \quad  +\mu^{-1}\vect{v}^\top\mat \Gamma_p \mat{M}\vect{v} +\mu^{-1}\vect{u}^\top\mat \Gamma_p \dot{\mat{M}}\vect{v}+\mu^{-1}\vect{u}^\top \mat \Gamma_p \mat{M}\dot{\vect{v}} \label{eq:lyapunov_candidate_derivative_a} \\
 & = -\frac{1}{2}\mu\vect{v}^\top \mat Q \vect{v} -\frac{1}{2}\mu^{-1}\vect{u}^\top \mat Q \vect{u} + \mu^{-1}\vect{v}^\top \mat \Gamma_p \mat{M}\vect{v} \nonumber \\
& \;  +\mu^{-1}\vect{u}^\top \mat \Gamma_p \mat{C}^\top\vect{v}  -\vect{v}^\top \mat \Gamma_p \vect{e} - \mu^{-1}\vect{u}^\top \mat \Gamma_p \vect{e} \label{eq:lyapunov_candidate_derivative_b}
\end{align}
We obtain \eqref{eq:lyapunov_candidate_derivative_b} by substituting in $\mat{M}\dot{\vect{v}}$ from \eqref{eq:euler_lagrange_general_system_error} noting that $\dot{\mat{M}}-2\mat{C}$ is skew-symmetric, or equivalently $\dot{\mat{M}} = \mat{C}+\mat{C}^\top$. Using the properties of the Euler-Lagrange system (\ref{prty:EL_M_symPD} to \ref{prty:EL_g_bound}), the following upper bound on $\dot{V}$ is obtained.
\begin{align}\label{eq:lyapunov_candidate_derivative_bound}
\dot{V} & \leq  - (\frac{1}{2}\mu\lambda_{\min}(\mat Q) - \mu^{-1}k_{\overline{M}} \bar \gamma) \Vert \vect v\Vert^2 -\frac{1}{2}\mu^{-1}\lambda_{\min}(\mat Q) \Vert \vect u\Vert^2\nonumber \\
& \; + \mu^{-1}k_C \bar \gamma  \Vert \vect{v} \Vert^2 \Vert \vect u \Vert + \bar \gamma \Vert \vect{v} \Vert \Vert \vect{e} \Vert + \mu^{-1}\bar \gamma \Vert \vect{u}\Vert \Vert \vect{e} \Vert 
\end{align}
Using the bound on $\Vert \vect e \Vert$ computed in \eqref{eq:e_norm_bound}, we then evaluate \eqref{eq:lyapunov_candidate_derivative_bound} to be
\begin{align}
\dot{V} & \leq  - (\frac{1}{2}\mu\lambda_{\min}(\mat Q) - \mu^{-1}k_{\overline{M}} \bar \gamma) \Vert \vect v\Vert^2 -\frac{1}{2}\mu^{-1}\lambda_{\min}(\mat Q) \Vert \vect u\Vert^2 \nonumber \\
& \; + \mu^{-1}k_C \bar \gamma  \Vert \vect{v} \Vert^2 \Vert \vect u \Vert + \mu^{-1}\beta_1 \bar \gamma  \Vert \mathcal{L}_{22}\Vert \Vert \vect u \Vert \Vert \vect{v}  \Vert \nonumber\\
& \;  + \beta_2 \bar \gamma \Vert \mathcal{L}_{22} \Vert \Vert \vect{v}  \Vert^2 +  \mu^{-2}\beta_1 \bar\gamma \Vert \mathcal{L}_{22}\Vert \Vert \vect u \Vert^2 \nonumber \\ 
& \;  + \mu^{-1} \beta_2 \bar \gamma \Vert \mathcal{L}_{22} \Vert \Vert \vect{u}\Vert \Vert \vect v \Vert + \mu^{-1} \bar\omega(t)  \bar \gamma \Vert \vect{u}\Vert + \bar\omega(t)\bar \gamma \Vert \vect{v}  \Vert \label{eq:Vdot_derivative_boundb} \\
& = -\mu^{-1}\Bigg[ A_1 \Vert \vect u \Vert^2 + A_2 \Vert \vect v \Vert^2 - A_3 \Vert \vect v \Vert^2 \Vert \vect u \Vert - A_4 \Vert \vect v \Vert \Vert \vect u \Vert \nonumber \\
& \; - A_5 \Vert \vect u \Vert - A_6 \Vert \vect v \Vert  \Bigg] := - \mu^{-1} p(\Vert \vect u \Vert, \Vert \vect v \Vert) \label{eq:p_uv_def}
\end{align}
where $A_1(\mu) = \lambda_{\min}(\mat Q)/2 - \mu^{-1}\bar\gamma\beta_1\Vert \mathcal{L}_{22}\Vert$, $A_2(\mu) = \mu^2\lambda_{\min}(\mat Q)/2 - \bar\gamma(k_{\overline{M}} + \mu\beta_2 )$, $A_3 = k_C \bar \gamma$, $A_4 = \Vert \mat\Gamma \mathcal{L}_{22} \Vert\bar\gamma (\beta_1 + \beta_2)$, $A_5(t) = \bar\gamma\bar\omega(t)$, and $A_6(\mu,t) = \mu\bar\gamma\bar\omega(t)$. By designing $\beta_1$ such that 
\begin{equation*}
\beta_1 < \frac{\mu_5^*\lambda_{\min}(\mat Q)}{2\bar\gamma\Vert \mathcal{L}_{22}\Vert}
\end{equation*}
then $A_1(\mu) > 0$ for any $\mu \geq \mu_5^*$. Observe that $A_2(\mu_5^*) > 0$ if $(\mu_5^*)^2\lambda_{\min}(\mat Q)/2 - \gamma k_{\overline{M}} - \mu_5^* \beta_2 > 0$. Rearranging for $\beta_2$, this is implied by
\begin{equation}\label{eq:beta2_req}
\beta_2 < \frac{\mu_5^* \lambda_{\min}(\mat Q) }{2} - \frac{\bar\gamma k_{\overline{M}} }{\mu_5^*}
\end{equation}
and note that any $\beta_2$ satisfying \eqref{eq:beta2_req} continues to satisfy \eqref{eq:beta2_req} for any $\mu \geq \mu_5^*$. If the right hand side of \eqref{eq:beta2_req} is negative, it is still possible to ensure that $A_2(\mu_5^*) > 0$ by increasing the size of $\mu_5^*$ and setting $\beta_2$ sufficiently small because the coefficient of $\mu^2$ in $A_2$ is positive. Lastly, observe that as $\mu \to \infty$ then $A_1(\mu) \to \lambda_{\min}(\mat Q)/2$, $A_2(\mu) = \mathcal{O}(\mu^2)$ and $A_6(\mu) = \mathcal{O}(\mu)$. Notice that $\bar\omega(t)$ decays to $0$ exponentially fast. In other words, the coefficients $A_5(t)$ and $A_6(\mu,t)$ decay to zero exponentially fast. 

\emph{Part 2:} We now show that the trajectories of the system are bounded for all time by carefully designing $\mu$. For \emph{Part 2} and \emph{Part 3} of the proof, a diagram is included in Fig.~\ref{fig:stability_diag} to aid in the explanation of the proof. Notice that $p(\Vert \vect u \Vert, \Vert \vect v \Vert)$ in \eqref{eq:p_uv_def} is of the same form as $h(x,y)$ in \textbf{Corollary}~\ref{cor:bounded_pd_function_2} with $\Vert \vect u \Vert = x$, $\Vert \vect v \Vert = y$ and $A_1(\mu) = a$, $A_2(\mu) = b$, $A_3 = c$, $A_4 = d$, $A_5(0) = e$ and $A_6(\mu,0) = f$. Note that $A_5(t_1) > A_5(t_2)$ and $A_6(\mu,t_1) > A_6(\mu,t_2)$ for any $t_1 < t_2$. Because of this, we proceed using $A_5(0), A_6(\mu,0)$: any $A_2(\mu)$ satisfying the inequalities on $b$ in \textbf{Corollary}~\ref{cor:bounded_pd_function_2} for $t = 0$ will continue to satisfy the inequalities for $t > 0$. This will be come clear in the sequel. We now use the values $\mathcal{X}, \mathcal{Y}$ computed in subsection~\ref{subsec:IC_bound}. Choose $\vartheta_0 > \mathcal{X} - \mathcal{X}_1$ and $\varphi_0 > \mathcal{Y}-\mathcal{Y}_1$, and ensure that $\mathcal{X}-\vartheta_0, \mathcal{Y}-\varphi_0 > 0$. Note the fact that $\mathcal{X}\geq \mathcal{X}_1$ and $\mathcal{Y}\geq \mathcal{Y}_1$ implies $\vartheta_0, \varphi_0 > 0$. We assume without loss of generality that $\mathcal{X}$ found in Section~\ref{subsec:IC_bound} is such that $\mathcal{X} - \vartheta_0 > A_5(0)/A_1(\mu)$. If this inequality were not satisfied, one would replace $\mathcal X$ by a $\bar{\mathcal{X}} > \mathcal{X}$ such that $\bar{\mathcal{X}}-\vartheta_0 > e/a$ and proceed with the stability proof using $\bar{\mathcal{X}}$. 

Define the sets $\mathcal{U}, \mathcal{V}$ and the region $\mathcal{R}$ as in \textbf{Corollary}~\ref{cor:bounded_pd_function_2} with $\Vert \vect u \Vert = x$, $\Vert \vect v \Vert = y$. Define further the sets $\bar{\mathcal{U}} = \{\Vert \vect u \Vert: \Vert \vect u \Vert > \mathcal{X}\}$ and $\bar{\mathcal{V}} = \{\Vert \vect v \Vert: \Vert \vect v \Vert > \mathcal{X}\}$. Define the compact region $\mathcal{S} = \mathcal{U} \cup \mathcal{V} \backslash \bar{\mathcal{U}} \cup \bar{\mathcal{V}}$ (refer to Fig.~\ref{fig:stability_diag} for details). Since $\mathcal{S} \subset \mathcal{R}$, there exists a $\mu_6^* \geq \mu_5^*$ such that $A_2(\mu_6^*)$ satisfies the requirement on $b$ in \textbf{Corollary}~\ref{cor:bounded_pd_function_2}, which ensures that $p(\Vert \vect u \Vert, \Vert \vect v \Vert) > 0$. This in turn implies that $\dot{V} < 0$ in $\mathcal{S}$. Lastly, define the region $\Vert \vect u(t) \Vert \in [0, \mathcal{X} - \vartheta_0)$ and $\Vert \vect v(t) \Vert \in [0, \mathcal{Y} - \varphi_0)$ as $\mathcal{T}$. In this part of the proof, we are trying to show that the trajectories of \eqref{eq:euler_Lagrange_nonautonomous_network} remain bounded for all time. For purposes of explanation, we therefore temporarily assume that $\mathcal{S}$ and $\mathcal{T}$ are time-invariant, as opposed to Fig.~\ref{fig:stability_diag}. In the latter \emph{Part 3}, we will discuss the time-varying nature of $\mathcal{S}(t)$ and $\mathcal{T}(t)$ and show that the boundedness arguments developed here continue to hold. 

We are now ready to show that the trajectory of the system \eqref{eq:euler_Lagrange_nonautonomous_network} remains in $\mathcal{T}\cup \mathcal{S}$ for all $t\geq 0$. We define $T_1$ as the infimum of time values for which either $\Vert \vect u(t) \Vert < \mathcal{X}$ or $\Vert \vect v(t) \Vert < \mathcal{Y}$ fail to hold. We show that the existence of $T_1$ creates a contradiction, and thus conclude that the bounds $\Vert \vect u(t) \Vert < \mathcal{X}$ or $\Vert \vect v(t) \Vert < \mathcal{Y}$ hold for all $t$. 

Observe that $\dot{V}$ may be sign indefinite in $\mathcal{T}$, which means if the trajectory of the system is in $\mathcal{T}$ (the blue region in Fig.~\ref{fig:stability_diag}) then $V(t)$ can increase. However, from \eqref{eq:lyapunov_upperbound} we obtain that in $\mathcal{T}$
\begin{align*}
V(t) & \leq \frac{1}{2}\lambda_{\max}(\mat Q) ({\mathcal{X} - \vartheta_0})^2 + \frac{1}{2}\bar\gamma(k_{\overline{M}}+\delta) ({\mathcal{Y} - \varphi_0})^2 \\
&\; + {\mu}^{-1}\bar\gamma(k_{\overline{M}}+\delta) (\mathcal{X} - \vartheta_0) (\mathcal{Y} - \varphi_0) := \mathcal{Z}
\end{align*}
One can easily verify that $\mathcal{Z} < \wh{V}^*$ because we selected $\vartheta_0,\varphi_0$ such that $\mathcal{X}_1 > \mathcal{X} - \vartheta_0$ and $\mathcal{Y}_1 > \mathcal{Y} - \varphi_0$. Note that any trajectory of \eqref{eq:euler_Lagrange_nonautonomous_network} beginning\footnote{
The definitions of $\mathcal{S}$ and $\mathcal{T}$ ensure that $\vect u(0), \vect v(0) \in \mathcal{S}\cup \mathcal{T}$ as evident in \eqref{eq:semi_global_barV_delta}.} in $\mathcal{S}\cup \mathcal{T}$ must satisfy $V(t) \leq \max \{\mathcal{Z}, V(0)\} < \wh{V}^*$ for $t \in [0,T_1]$. This is because $\dot{V} < 0$ in $\mathcal{S}$; any trajectory starting in $\mathcal{S}$ (respectively $\mathcal{T}$) has $V(t) < V(0)$ (respectively $V(t) < \mathcal{Z}$). If the trajectory leaves $\mathcal{T}$ and enters $\mathcal{S}$ at some $t$, consider the crossover point, which is in the closure of $\mathcal{T}$. By the virtue that $V$ is continuous, we have $V(t) < Z$ and by virtue of entering $\mathcal{S}$, $V(t+\delta) \leq V(t) < \mathcal{Z}$, for some arbitrarily small $\delta$. Define $\zeta = \lambda_{\min}(\mat M - \mu^{-1}\mat M \mat Q^{-1} \mat M)/2 $ and verify that $\zeta \geq \rho_3(\mu_3^*) > \rho_2(\mu^*_3)$. In accordance with \textbf{Lemma}~\ref{lem:W_bound} we have
\begin{equation}\label{eq:Y_contradict}
\Vert \vect v(T_1)\Vert_2 \leq \sqrt{\frac{V(T_1)}{\zeta}} < \sqrt{\frac{\wh{V}^*}{\zeta}} < \sqrt{\frac{\wh{V}^*}{\rho_2(\mu^*_3)}} = \mathcal{Y}
\end{equation}
Paralleling the argument leading to \eqref{eq:Y_contradict}, one can also show that $\Vert \vect u(T_1) \Vert < \mathcal{X}$. We omit this due to spatial limitations and similarity of argument. The existence of \eqref{eq:Y_contradict} and a similar inequality for $\Vert \vect u(T_1)\Vert$ contradict the definition of $T_1$. In other words, $T_1$ does not exist, and therefore $\Vert \vect u(t) \Vert < \mathcal{X}$ and $\Vert \vect v(t) \Vert < \mathcal{Y}$ for all $t$, as depicted in Fig~\ref{fig:stability_diag}. 

\emph{Part 3:} We now show that the leader-follower consensus objective is achieved. In order to do this, we firstly explore how the time-varying nature of $A_5(t)$ and $A_6(\mu,t)$ affects $\dot{V}$. As discussed in Fig.~\ref{fig:stability_diag}, $A_5(t)$ and $A_6(\mu, t)$ decay to zero exponentially fast due to the presence of $\bar\omega(t)$. Therefore, at $t = \infty$, $p(\Vert \vect u \Vert, \Vert \vect v \Vert)$ is of the form of $g(x,y)$ in \textbf{Lemma}~\ref{lem:bounded_pd_function} and is thus positive definite (and therefore $\dot{V} = -\mu^{-1} p$ is negative definite) for $\Vert \vect u \Vert \in [0, \mathcal{X}]$ and $\Vert \vect v \Vert \in [0,\infty)$. This is because $\mu_6^*$ as designed according to \textbf{Corollary}~\ref{cor:bounded_pd_function_2} also satisfies the requirements detailed in \textbf{Lemma}~\ref{lem:bounded_pd_function}. It is also straightforward to conclude that the sign indefiniteness of $\dot{V}(t)$ in $\mathcal{T}(t)$ arises due to the terms linear in $\Vert \vect u \Vert$ and $\Vert \vect v \Vert$ in \eqref{eq:Vdot_derivative_boundb}, i.e. the terms containing $\bar\omega(t)$, which are precisely the coefficients $A_5(t)$ and $A_6(\mu,t)$ in \eqref{eq:p_uv_def}.

Now that we have established that $A_5(t) \Vert \vect u\Vert$ and $A_6(\mu,t) \Vert \vect v \Vert$ give rise to the region $\mathcal{T}(t)$, we now establish the precise behaviour of $\mathcal{T}(t)$ and $\mathcal{S}(t)$ as functions of time. Now examine the inequalities on the coefficient $b$ as detailed in \textbf{Corollary}~\ref{cor:bounded_pd_function_2}, applied to $p(\Vert \vect u \Vert, \Vert \vect v \Vert)$ in \eqref{eq:p_uv_def}. We conclude that for a fixed $\mu_6^*$, the strictly monotonically decreasing nature of $A_5, A_6$ then implies that, for fixed $A_2(\mu_6^*)$, the region in which $p(\Vert \vect u \Vert, \Vert \vect v \Vert) > 0$ (respectively sign indefinite) as defined by $\mathcal{S}(t)$ (respectively $\mathcal{T}(t)$), is time-varying. \emph{Specifically, $\vartheta(t)$ and $\varphi(t)$ are strictly monotonically increasing.} Moreover, because $A_5 = A_6 = 0$ at $t = \infty$, we conclude that $\lim_{t\to \infty} \vartheta(t) = \mathcal{X}$ and $\lim_{t\to \infty} \varphi(t) = \mathcal{Y}$, at which point $\mathcal{T}(t) = [0,0]$.

It is straightforward to show that appropriate functions are given by $\vartheta(t) = -a_1e^{- a_2 t} + \mathcal{X}$ and $\varphi(t) = -b_1e^{- b_2 t} + \mathcal{Y}$. Here, $a_1, a_2, b_1, b_2$ are positive constants associated with $\mathcal{X}_0$, $\mathcal{Y}_0$ and the decay rate of $\bar\omega(t)$. Moreover, because $\vartheta(t), \varphi(t)$ are strictly monotonically increasing, it is easy to verify that $\mathcal{S}(t_1) \subset \mathcal{S}(t_2)$ and $\mathcal{T}(t_1) \supset \mathcal{T}(t_2)$ for any $t_1 < t_2$. These properties ensure that the boundedness arguments developed in \emph{Part 2} remain valid for time-varying $\mathcal{S}(t)$ and $\mathcal{T}(t)$ due to the nature of the time variation.

Now that we have established the behaviour of $\mathcal{S}(t)$ and $\mathcal{T}(t)$, we move on to show that leader-follower consensus is achieved. Suppose that at some $T_2$, the trajectory of the system \eqref{eq:euler_Lagrange_nonautonomous_network} leaves $\mathcal{T}(t)$ and does not enter $\mathcal{T}(t)$ for all $t \geq T_2$. In other words, for $t\geq T_2$, the trajectory is in $\mathcal{S}(t)$ (recall that in \emph{Part 2} we established that $\Vert \vect u(t) \Vert < \mathcal{X}$ and $\Vert \vect v(t) \Vert < \mathcal{Y}$ for all $t$). This is illustrated in Fig.~\ref{fig:stability_diag}. 

Firstly, consider the case where $T_2 = \infty$. From the form of $\vartheta(t), \varphi(t)$, we conclude that $\mathcal{T}(t)$ shrinks exponentially fast towards the origin $\vect u = \vect v = 0$. If $T_2 = \infty$ then there is some $T_3 < \infty$ such that the trajectory of the system \eqref{eq:euler_Lagrange_nonautonomous_network} is in $\mathcal{T}(t)$ for all $t\in [T_3, \infty)$. From the limiting behaviour of $\mathcal{T}(t)$, we conclude that the trajectory of \eqref{eq:euler_Lagrange_nonautonomous_network} also converges to the equilibrium $\Vert \vect u \Vert = \Vert \vect v \Vert = 0$, which implies that the leader-follower consensus objective has been achieved. Moreover, the convergence rate is exponential for $t\in [T_3,\infty)$.

Secondly, consider the case where $T_2$ is finite. Note that $T_2$ is initial condition dependent, but the initial condition set is bounded according to Assumption~\ref{assm:IC}. It follows that there exists a $\bar{T}_2$ \emph{independent of initial conditions} such that, for every initial condition satisfying Assumption~\ref{assm:IC}, $T_2 < \bar{T}_2 < \infty$. Define the time interval $t_p = [\bar{T}_2, \infty)$. Because the trajectory of the system \eqref{eq:euler_Lagrange_nonautonomous_network} is in $\mathcal{S}(t)$ for all $t\in t_p$ then $\dot{V}(t) < 0$ for all $t\in t_p$. Consider some arbitrary time $t_1 \in t_p$. We observe that $p(\Vert \vect u \Vert, \Vert \vect v \Vert) >0$ (i.e. positive definite) in the compact region $\mathcal{S}(t_1)$. One can therefore find a scalar $a_{1, t_1} > 0$ such that $p(\Vert \vect u \Vert, \Vert \vect v \Vert) \geq a_{1,t_1} \Vert [\vect u^\top, \vect v^\top ]^\top\Vert$ for all $\Vert \vect u \Vert, \Vert \vect v \Vert \in \mathcal{S}(t_1)$. Furthermore, by recalling that $A_5, A_6$ are positive and strictly monotonically decreasing, we conclude that $p(\Vert \vect u \Vert, \Vert \vect v \Vert, t_1) < p(\Vert \vect u \Vert, \Vert \vect v \Vert, t_2)$ for any $\vect u, \vect v$, and for any $t_1 \leq t_2$ where $t_1, t_2 \in t_p$. It follows that there exists some constant $\bar a_1 > 0$ such that $p(\Vert \vect u \Vert, \Vert \vect v \Vert) \geq \bar a_{1} \Vert [\vect u^\top, \vect v^\top ]^\top\Vert$ for all $\vect v, \vect u$ in $\mathcal{S}(t)$, for all $t \in t_p$. This implies that, in $\mathcal{S}(t)$ we have $\dot{V} \leq -\bar a_{1} \Vert [\vect u^\top, \vect v^\top]^\top \Vert < 0$ for all $t\in t_p$. The eigenvalues of the constant matrices $\mat L_\mu$ and $\mat N_\mu$ (the matrices introduced in subsection~\ref{subsec:IC_bound}) are finite and strictly positive. Our earlier conclusion that $\mat L_\mu > \mat G > \mat N_\mu$ for all $\mu \geq \mu_3^*$ then implies that the eigenvalues of $\mat G$ (which vary with $\vect q(t)$) are upper bounded away from infinity and lower bounded away from zero. It follows that there exist scalars $a_2, a_3 > 0$ such that $a_2 \Vert [\vect u^\top, \vect v^\top ]^\top\Vert \leq V(t) \leq a_3 \Vert [\vect u^\top, \vect v^\top]^\top \Vert$. This implies that $\dot{V}(t) \leq - \psi V(t)$ in $\mathcal{S}(t)$ for $t\in t_p$ where $\psi = \bar a_{1}/a_3$. From this, we conclude that $V$ decays exponentially fast to zero, with a minimum rate $e^{-\psi t}$, for $t \in t_p$. Since $V$ is positive definite in $\vect u, \vect v$, this implies that $\Vert [\vect u^\top, \vect v^\top]^\top \Vert$ decays to zero exponentially fast for $t \in t_p$, and the leader-follower consensus objective is achieved.

\begin{figure}[t]
\centering{
\resizebox{\columnwidth}{!}{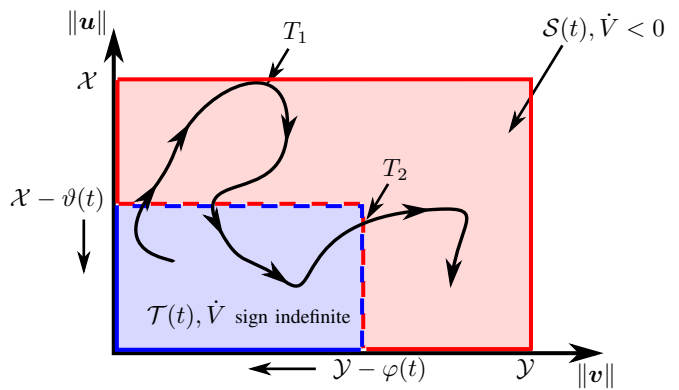}
\caption{Diagram for proof of \textbf{Theorem}~\ref{theorem:directed_main_MI}. The red region is $\mathcal{S}(t)$, in which $\dot{V}(t) < 0$ for all $t \geq 0$. The blue region is $\mathcal{T}(t)$, in which $\dot{V}(t)$ is sign indefinite. A trajectory of \eqref{eq:euler_Lagrange_nonautonomous_network} is shown with the black curve. At $t = T_1$, it is shown in \emph{Part 2} that the trajectory of \eqref{eq:euler_Lagrange_nonautonomous_network} is such that $
\Vert \vect u(T_1) \Vert < \mathcal{X}, \Vert \vect v(T_1) \Vert < \mathcal{Y}$ and thus the trajectory does not leave $\mathcal{S}(t)$. 
The sign indefiniteness of $\dot{V}(t)$ in $\mathcal{T}(t)$ arises due to the terms linear in $\Vert \vect u \Vert$ and $\Vert \vect v \Vert$ in \eqref{eq:Vdot_derivative_boundb}, i.e. the terms containing $\bar\omega(t)$ (coefficients $A_5(t)$ and $A_6(\mu,t)$ in \eqref{eq:p_uv_def}). Because $\bar\omega(t)$ goes to zero at an exponential rate, so do the coefficients $A_5(t)$ and $A_6(\mu,t)$. Examining the inequalities detailed in \textbf{Corollary}~\ref{cor:bounded_pd_function_2} as applied to $p(\Vert \vect u \Vert, \Vert \vect v \Vert)$ in \eqref{eq:p_uv_def}, it is straightforward to conclude that for a fixed $\mu_6^*$, the exponential decay of $A_5, A_6$ implies that the region $\mathcal{T}(t)$ shrinks towards the origin at an exponential rate. In other words, $\vartheta(t)$ and $\varphi(t)$ monotonically increase until $\vartheta(t) = \mathcal{X}$ and $\varphi(t) = \mathcal{Y}$, at which point $\mathcal{T}(t) = [0,0]$. This corresponds to the dotted red and blue lines, which show, respectively, the time-varying boundaries of $\mathcal{S}(t)$ and $\mathcal{T}(t)$. The solid red and blue lines show respectively, the boundaries of $\mathcal{S}(t)$ and $\mathcal{T}(t)$, which are time-invariant. Exponential convergence to the leader-follower objective is discussed in \emph{Part 3} making using of $T_2$.
}
\label{fig:stability_diag}
}
\end{figure}


\end{document}